\documentclass[12pt,a4paper]{article}
\usepackage[utf8]{inputenc}
\usepackage[affil-it]{authblk}
\usepackage{amsmath}
\usepackage{amsthm}
\usepackage{bbm}
\usepackage[top=3cm,left=3cm,right=3cm,bottom=3cm]{geometry}
\usepackage{txfonts}
\usepackage{amsfonts}
\usepackage{amssymb}
\usepackage{enumerate}
\usepackage{enumitem}
\usepackage{verbatim}
\usepackage{graphicx}
\usepackage{tikz}
\usepackage{tikz-cd}

\allowdisplaybreaks

\newcommand{\ddim}{{d}} 
\newcommand{\newtonian}[1]{{\Gamma\left({#1}\right)}}

\newcommand{\NN}{\mathbb N}
\newcommand{\IN}{\mathbb Z}
\newcommand{\RN}{\mathbb R}
\newcommand{\CN}{\mathbb C}

\newcommand{\mi}{\mathrm i}

\newcommand{\HS}{\mathcal H}
\newcommand{\SP}[2]{\left\langle {#1} , {#2} \right\rangle}

\newcommand{\norm}[1]{\left|\left|{#1}\right|\right|}

\newcommand{\symp}[2]{\left[{#1},{#2}\right]}

\newcommand{\sympf}[2]{\omega\left({#1},{#2}\right)}

\newcommand{\Lp}[2]{\mathcal{L}^{{#1}}_{{#2}}}
\newcommand{\Lpn}[3]{\norm{{#1}}_{\Lp{{#2}}{{#3}}}}

\newcommand{\Wp}[2]{\mathcal{W}^{{#1}}_{{#2}}}
\newcommand{\Wpn}[3]{\norm{{#1}}_{\Wp{{#2}}{{#3}}}}
\newcommand{\Cp}[2]{\mathcal{C}^{{#1}}_{{#2}}}

\newcommand{\Mp}[2]{\mathcal{M}^{{#1}}_{{#2}}}
\newcommand{\Mpn}[3]{\norm{{#1}}_{\Mp{{#2}}{{#3}}}}

\newcommand{\id}{\textrm{id}}
\newcommand{\unity}{\ensuremath{\mathbbm{1}}}
\newcommand{\intd}{\mathrm d}
\newcommand{\abs}[1]{\left|{#1}\right|}
\newcommand{\inv}[1]{{#1}^{-1}}
\newcommand{\convol}[2]{\left({#1}*{#2}\right)}

\newcommand{\Cm}[2]{\mathcal C^{#1}\left({#2}\right)}

\newcommand{\fourier}{\mathcal F}
\newcommand{\VF}[1]{\hat{{#1}}}

\newcommand{\HF}{\mathcal{H}}

\newcommand{\HVF}{\mathcal{H}_{\text{Vl.}}}
\newcommand{\HHF}{\mathcal{H}_{\text{Ht.}}}
\newcommand{\RX}{{\RN^{\ddim}_{\bx}}}
\newcommand{\RXN}{\RN^{\ddim N}_{\vec{\bx}}}
\newcommand{\RV}{{\RN^{\ddim}_{\bv}}}
\newcommand{\RXi}{{\RN^{\ddim}_{\xi}}}
\newcommand{\RZ}{{\RN^{2\ddim}_{\bz}}}
\newcommand{\RZN}{{\RN^{2\ddim N}_{\vec{\bz}}}}
\newcommand{\RZH}{{\RN^{2\ddim}_{\VF{\bz}}}}
\newcommand{\RZHN}{{\RN^{2\ddim N}_{\vec{\VF{\bz}}}}}

\newcommand{\RNp}{\RN_{\geq 0}}
\newcommand{\naX}{\nabla_{\bx}}
\newcommand{\naV}{\nabla_{\bv}}

\newcommand{\naXi}{\nabla_{\xi}}

\newcommand{\Del}[1]{\mathrm D^{#1}}

\newcommand{\bx}{\textbf x}
\newcommand{\bv}{\textbf v}

\newcommand{\bz}{\textbf z}

\newcommand{\ba}{\textbf a}

\newcommand{\twovec}[2]{\left(\begin{array}{c} {#1} \\ {#2} \end{array}\right)}

\newcommand{\bound}[2]{\mathfrak{b}_{{#1}}^{{#2}}}

\makeindex

\begin{document}


\theoremstyle{plain}
\newtheorem{thm}{Theorem}[section]
\newtheorem{lem}[thm]{Lemma}
\newtheorem{cor}[thm]{Corollary}
\newtheorem{prop}[thm]{Proposition}
\newtheorem{conj}[thm]{Conjecture}

\theoremstyle{definition}
\newtheorem{defn}[thm]{Definition}

\theoremstyle{remark}
\newtheorem{rmk}[thm]{Remark}
\newtheorem{exam}[thm]{Example}

\title{A mean field limit for the \\ Hamiltonian Vlasov system}
\author{R.A. Neiss
	\thanks{
		Electronic address: \texttt{rneiss@math.uni-koeln.de}
	}
}
\affil{
	Universität zu Köln, Mathematisches Institut, \\
	Weyertal 86-90, 50931 Köln, Germany
}
\author{P. Pickl
	\thanks{
		Electronic address: \texttt{peter.pickl@dukekunshan.edu.cn}
	}
}
\affil{
	Duke Kunshan University, \\
	8 Duke Avenue, Kunshan City, \\
	Jiangsu Province, PR China 215316
}
\date{November 29, 2018}

\maketitle

\begin{abstract}
\noindent 
The derivation of effective equations for interacting many body systems has seen a lot of progress in the recent years. While dealing with classical systems, singular potentials are quite challenging \cite{HaurayJabin,lazarovicipickl} comparably strong results are known to hold for quantum systems \cite{knowlespickl}.
In this paper, we wish to show how techniques developed for the derivation of effective descriptions of quantum systems can be used for classical ones. While our future goal is to use these ideas to treat singularities in the interaction, the focus here is to present how quantum mechanical techniques can be used for a classical system and we restrict ourselves to regular two-body  interaction potentials.  In particular we compute a mean field limit for the Hamilton Vlasov system in the sense of \cite{froehlichknowlesschwarz,neiss} that arises from  classical dynamics. The structure reveals strong analogy to the bosonic quantum mechanical ensemble of the many-particle Schrödinger equation and the Hartree equation as its mean field limit \cite{pickl}. 
\end{abstract}

\section{Introduction}

The Vlasov equation effectively describes the collective behaviour of many particle systems by reducing the information to a distribution function on the one-particle phase space.  The system has been under active research for many decades and there are lots of results about it regarding various setups, such as interactions, symmetries, dimensions, boundary conditions \cite{Horst,Horst3,LionsPerthame,Loeper,Schaeffer}. One of the central questions is to derive the equation from many body Newtonian dynamics \cite{Boers,Dobrushin,HaurayJabin,Fournier,Kiessling,Neunzert,NeunzertWick,braunhepp,Sznitman,SpohnBook} however, proving the validity of the Vlasov equation in effectively describing the dynamics of a classical gas with Coulomb interaction is still an open problem.

In this manuscript we wish to show how quantum mechanical techniques can be used to derive such effective equations for classical systems. 
Finding a decent Hamiltonian formulation for such a classical system is interesting on its own, but it also might be the basis for generalizing present results with respect to the singular behaviour of the interaction and/or to adapt tools from the field of finite dimensional Hamiltonian systems. \\

\noindent Among older formal approaches \cite{marsdenweinstein,morrison,yemorrisoncrawford} that focus more on the algebraic structure of the problem, such as symplectic foliation, a more recent idea by Fröhlich, Knowles, and Schwarz \cite[Sec.2]{froehlichknowlesschwarz} is elaborated in \cite[Sec.2]{neiss} for a general Vlasov setting. Therein, for any Vlasov system with a phase space consisting of non-negative distribution functions $f\in\Lp{1}{\bz}$, a symplectic manifold and a Hamiltonian equation for corresponding complex-valued $\alpha\in\Lp{2}{\bz}$ are defined, such that $f=\abs{\alpha}^2$ yields a solution of the classical Vlasov system again. This provides a widely applicable framework that seems more accessible from a rigorous PDE point of view. The Hamiltonian system of these $\alpha$ is referred to as \textit{Hamilton Vlasov system} in this paper. \\

\noindent While the purpose of defining a Hamiltonian structure is not obvious at first sight, we shall show in  this paper that it might be helpful to prove a mean field limit for the Hamilton Vlasov system itself. Note, that the existing techniques for proving mean field limits of classical Vlasov systems, e.g. \cite{braunhepp,lazarovicipickl}, cannot be lifted easily to the $\Lp{2}{}$ setup, because they are mainly concerned with the convergence of marginal distributions, a concept that has no immediate $\Lp{2}{}$ counterpart.

Our paper focuses on the special case of a Newtonian system with a regular two-body interaction force. For this an underlying many-particle system is constructed, and, in strong analogy to bosonic quantum systems, a mean field limit is proven. The effective equation obtained thereby exactly recovers the Hamilton Vlasov system. Although the physical interpretation of the Hamiltonian formulation might be questionable, it serves as a technical tool for proving the mean field limit using $\Lp{2}{}$ methods. It is a peculiar pseudo quantum mechanical system, structurally very similar to the many-particle bosonic Schrödinger/ Hartree ensemble. Nevertheless, significant differences arise, most importantly a kinetic term that is now hyperbolic rather than elliptic.

At this point, it is already worth noting, that by the non-injective map $\alpha \mapsto \abs{\alpha}^2$, more thoroughly explored in Section \ref{sec:mf:eqn-hierarchy}, any mean field limit of the Hamilton Vlasov system also yields some limit for the classical Vlasov setting. Therefore, the results even have an impact on the significant problem of Vlasov mean field limits. \\

\noindent The paper is structured as follows. Section \ref{sec:mf:eqn-hierarchy} recalls and defines the relevant equations and setup, and puts them in hierarchical order. Section \ref{sec:mf:regular-global-well-posedness} proves global well-posedness of the $\Lp{2}{}$ systems for a regular potential and Section \ref{sec:mf:mean-field-limit} finally proves the mean field limit for these potentials by similar means as in the bosonic quantum mechanical setting of \cite{pickl}.

It is worth noting, that if uniform regularity of the solutions of the microscopic equations $\alpha_N$ holds, the method even provides access up to interaction potentials satisfying $\nabla\Gamma\in\Lp{2}{\bx}$. This includes cases of mild singularities without cutoff. See Remark \ref{rmk:mf:singular-potentials} for a detailed discussion.
\section{Hierarchy of equations}

\label{sec:mf:eqn-hierarchy}

\noindent Throughout the paper let $\ddim\in\NN$ denote the dimension  of the underlying physical space, for example $\ddim=3$. The physical interaction potential $\Gamma:\RX\to\RN$ shall be even and only depend on the difference of the positions of two particles. 

\subsection{Overview of dynamical systems}

At this point, we want to recall the various ensembles considered in this paper and introduce our notation.

\subsubsection{Classical Vlasov system}

\noindent The classical Vlasov system is the mean field theory for a microscopic many-particle system of indistinguishable particles. Its Hamiltonian equations of motion for $N$ point masses of equal mass $1/N$ with two-body interaction are considered. We assume that the initial condition, given by a point in the $N$-particle phase space, is at random. The respective probability density is assumed to be i.i.d. and continuous, more precisely it is given by the product state $\mathring{f}_N\equiv \mathring{f}^{\otimes N}$. The equation of motion then is the \textit{Liouville equation}, i.e.,
\begin{align}
\label{eqn:mf:liouville-f} \tag{Lv\textsuperscript{N}}
\partial_tf_N(t,\bz_1,\dots,\bz_N) =& \symp{H_N}{f_N(t)}(\bz_1,\dots,\bz_N) \\
\nonumber
=&~ \sum_{m=1}^{N} \left(-\bv_m\cdot\nabla_{\bx_m} + \left(\frac{1}{N-1} \sum_{n\neq m} \nabla\Gamma(\bx_m-\bx_n)\right)\cdot \nabla_{\bv_m}\right) f_N(t,\bz_1,\dots,\bz_N),
\end{align}
where $[\cdot,\cdot]$ stands for the Poisson bracket on $\RZN$  and $H_N$  for the many-particle Hamiltonian given by
\begin{equation}
\label{eqn:mf:n-hamiltonian} \tag{Ham\textsuperscript{N}}
H_N(\bz_1,\dots,\bz_N) \equiv \sum_{m=1}^{N} \frac{\abs{\bv_m}^2}{2} + \frac{1}{2(N-1)} \sum_{\substack{m,n=1\\ m\neq n}}^{N} \Gamma(\bx_m-\bx_n).
\end{equation}
The rescaled energy is given by the  expectation value of the $N$ particle Hamiltonian w.r.t. the probability density $f_N$ on $\RZN$, i.e.,
\begin{align*}
\HF_N(f_N) = \int_{\RZN} H_N(\vec{\bz})~ f_N(\vec{\bz})~ \intd\vec{\bz}, \quad\text{where}\quad \vec{\bz} = (\bz_1,\dots,\bz_N).
\end{align*} \\

\noindent For any $f\in\Lp{1}{\bz}$ on the one particle phase space $\RZ=\RX\times\RV$ the classical Vlasov  energy functional is given by the energy functional defined for the many-particle evaluated for the product distribution $f^{\otimes N}$, i.e.,
\begin{align*}
\HF(f) =&~ \lim_{N\to\infty} \frac{1}{N} \HF_N\left(f^{\otimes N}\right) = \lim_{N\to\infty} \frac{1}{N} \int_{\RZN} H_N(\vec{\bz})~ \prod_{m=1}^{N} f(\bz_m)~ \intd\vec{\bz} \\
=&~ \int_{\RZ} \frac{\abs{\bv}^2}{2}~ f(\bz)~ \intd\bz + \frac{1}{2} \int_{\RZ\times\RZ} f(\bz_1)~ \Gamma(\bx_1-\bx_2)~ f(\bz_2)~ \intd(\bz_1,\bz_2),
\end{align*}
and the respective equation of motion is given by the \textit{Vlasov equation}
\begin{align}
\label{eqn:mf:vlasov} \tag{Vl}
\partial_t f(t,\bz) =&~ \symp{\frac{\abs{\bv}^2}{2}+\convol{\Gamma}{f(t)}(\bx)}{f(t)}(\bz) \\
\nonumber
=&~ -\bv\cdot\naX f(t,\bz) + \convol{\nabla\Gamma}{f(t)}(\bx) \cdot \naV f(t,\bz)\;.
\end{align}
Again,  $\symp{\cdot}{\cdot}$ stands for the  Poisson bracket, here on $\RZ$. \\

\noindent There exist numerous results on the limit $N\to\infty$ and the convergence towards the Vlasov description. A classical result is that the $k$ particle marginals for fixed $k$ of the joined distribution $f_N(t)$ tend to the product distribution $f(t)^{\otimes k}$, where $f(t)$ solves the Vlasov equation \cite{braunhepp}. Alternate results consider more carefully the difference of solution maps of the autonomous Hamiltonian system and the product of characteristics of the Vlasov system \cite{lazarovicipickl}, which is particularly useful for singular potentials, such as Coulomb with $N$-dependent cutoff length. \\

\subsubsection{Hamilton Vlasov system}

A Hamiltonian formulation of the Vlasov-Liouville ensemble is possible by the following construction. On the symplectic manifold of complex valued $\Lp{2}{\bz}$ functions $\alpha$ with symplectic form $\sympf{\cdot}{\cdot}\equiv \Im\SP{\cdot}{\cdot}$ the new Vlasov Hamiltonian \cite[Sec.2.2]{neiss} is constructed from the first functional derivative of the energy functional. To be more precise:
\begin{align*}
&~\HVF(\alpha) \equiv~ \frac{1}{2\mi} \Del{1}\HF\left(\abs{\alpha}^2\right)\left[\symp{\bar{\alpha}}{\alpha}\right] \\
=&~ \frac{1}{2\mi} \int_{\RZ} \frac{\abs{\bv}^2}{2}~ \symp {\bar{\alpha}} {\alpha}(\bx,\bv)~ \intd\bz + \frac{1}{2\mi} \int_{\RZ\times\RZ} \abs{\alpha(\bz_1)}^2~ \newtonian{\bx_1-\bx_2}~ \symp {\bar{\alpha}} {\alpha}(\bz_2)~ \intd(\bz_1,\bz_2).
\end{align*}
We call the respective equation of motion the \textit{Hamiltonian Vlasov equation}, it is given by
\begin{equation}
\label{eqn:mf:hamilton-vlasov} \tag{HVl}
\partial_t\alpha(t,\bz) = \symp{\frac{\abs{\bv}^2}{2}+\convol{\Gamma}{\abs{\alpha(t)}^2}(\bx)}{\alpha(t)}(\bz) - \convol{\Gamma}{\symp{\bar{\alpha}(t)}{\alpha(t)}}(\bx)~ \alpha(t,\bz),
\end{equation}
where the right hand side is the Hamiltonian vector field of $\HVF$ for the prescribed symplectic structure. 

Note that $f=\abs{\alpha}^2$, i.e., $\alpha$ is somehow the square root of the density. Nevertheless, the connection is quite  complex since the choice of the Hamiltonian functional $\HVF$ is not unique and influences the phase of $\alpha$ \cite[Sec.2.3]{neiss}. The choice we make here is most convenient from the technical point of view. \\

\noindent Similarly, on the $N$ particle distributions, we obtain the corresponding $N$-body  Hamiltonian functional
\begin{equation*}
\HF_{N,\text{Vl.}}(\alpha_N) \equiv \frac{1}{2\mi} \Del{1}\HF_N\left(\abs{\alpha_N}^2\right)\left[\symp{\bar{\alpha}_N}{\alpha_N}\right] = \frac{1}{2\mi} \int_{\RZN} H_N(\vec{\bz})~ \symp{\bar{\alpha}_N}{\alpha_N}(\vec{\bz})~ \intd\vec{\bz}.
\end{equation*}
This functional yields for $\alpha_N$ the \textit{Hamiltonian Liouville equation}
\begin{equation}
\label{eqn:mf:liouville-a} \tag{HLv\textsuperscript{N}}
\partial_t\alpha_N(t,\bz_1,\dots,\bz_N) = \symp{H_N}{\alpha_N(t)}(\bz_1,\dots,\bz_N).
\end{equation}
Note that by the linearity of the energy functional $\HF_N$, this equation is identical to $\eqref{eqn:mf:liouville-f}$.

\subsubsection{Pseudo QM system}

In order to prove the validity  of a mean field limit, it turns out to be useful to transform the Hamilton Vlasov system by  Fourier transformation in the velocity coordinate.

We denote the Fourier transform in the velocity coordinate by
\begin{equation}
\label{eqn:mf:velocity-fourier} \tag{FT}
\fourier: \Lp{2}{\bz} \to \Lp{2}{\VF{\bz}}, \quad \VF{\alpha}(\bx,\xi) \equiv \left(\fourier\alpha\right)(\bx,\xi) \equiv \frac{1}{\left(2\pi\right)^{\frac{\ddim}{2}}} \int_{\RV} \alpha(\bx,\bv)~ e^{-\mi\bv \cdot \xi}~ \intd\bv,
\end{equation}
in particular the conjugate variable of $\bv$ will consistently be denoted by $\xi\in\RXi$, further $\VF{\bz}\equiv(\bx,\xi)$. Of course, the (partial) Fourier transform is $\Lp{2}{}$-isometric and it acts nicely on Sobolev spaces. \\

\noindent The Hamiltonians $\HVF$ and $\HF_{N,\text{Vl.}}$ can be conveniently expressed in $\VF{\alpha}$ instead of $\alpha$, yielding new equations of motion for $\VF{\alpha}$, $\VF{\alpha}_N$ respectively. Introducing the notation
\begin{equation*}
\VF{V}(\bx,\xi)\equiv -\nabla\newtonian\bx\cdot\xi,
\end{equation*}
integrating by parts, and recalling the Plancherel Theorem for the Fourier transform a.e. in $\bx$, one computes
\begin{align*}
\HVF(\alpha) =&~ \frac{1}{2\mi} \int_{\RZ} \frac{\abs{\bv}^2}{2}~ \symp {\bar{\alpha}} {\alpha}(\bx,\bv)~ \intd\bz \\
&+ \frac{1}{2\mi} \int_{\RZ} \int_{\RZ} \abs{\alpha(\bz_1)}^2~ \newtonian{\bx_1-\bx_2}~ \symp {\bar{\alpha}} {\alpha}(\bz_2)~ \intd\bz_1~ \intd\bz_2 \\
\stackrel{\text{P.I.}}{=}&~ \frac{1}{2\mi} \int_{\RZ} \naX\alpha(\bx,\bv) \cdot \bv~ \bar{\alpha} (\bx,\bv)~ \intd\bz \\
&+ \frac{1}{2\mi} \int_{\RZ} \int_{\RZ} \abs{\alpha(\bz_1)}^2~ \bar{\alpha}(\bz_2)~ \nabla\newtonian{\bx_1-\bx_2}~ \naV\alpha(\bz_2)~ \intd\bz_1~ \intd\bz_2 \\
\stackrel{\text{F.T., Planch.}}{=}&~ -\frac{1}{2} \int_{\RZ} \naX\VF{\alpha}(\bx,\xi) \cdot \naXi\bar{\VF{\alpha}} (\bx,\xi)~ \intd\hat{\bz} \\
&+ \frac{1}{4} \int_{\RZ} \int_{\RZ} \abs{\VF{\alpha}(\hat{\bz}_1)}^2~ \abs{\hat{\alpha}(\hat{\bz}_2)}^2~ \nabla\newtonian{\bx_1-\bx_2} \cdot \left(\xi_2-\xi_1\right)~ \intd\hat{\bz}_1~ \intd\hat{\bz}_2 \\
=& -\frac{1}{2} \int_{\RZ} \naX\VF{\alpha}(\bx,\xi) \cdot \naXi\bar{\VF{\alpha}} (\bx,\xi)~ \intd\hat{\bz} + \frac{1}{4} \int_{\RZ} \hat{\alpha}(\hat{\bz})~ \convol {\hat{V}} {\abs{\hat{\alpha}}^2} (\hat{\bz})~ \bar{\VF{\alpha}}(\hat{\bz})~ \intd\hat{\bz} \\
\stackrel{\text{P.I.}}{=}&~ \frac{1}{2} \SP{\VF{\alpha}}{\left(\naX\cdot\naXi + \frac 12 \convol{\hat{V}}{\abs{\VF{\alpha}}^2}\right)~ \VF{\alpha}}_{\Lp{2}{\hat{\bz}}} \equiv \HHF(\VF{\alpha}),
\end{align*}
which is structurally intriguingly close to the Hartree energy functional. Nevertheless, there are important differences. The kinetic term is now hyperbolic instead of elliptic and the underlying space is $\RZH$ instead of $\RX$.

From the first derivative of this functional, the Hamiltonian vector field and the corresponding \textit{Hamilton Hartree equation} is computed to be
\begin{align}
\label{eqn:mf:hamilton-hartree} \tag{HHt}
\mi~\partial_t\hat{\alpha}(t,\hat\bz) = \left(\naX\cdot\naXi + \convol{\hat{V}}{\abs{\hat{\alpha}(t)}^2}(\hat{\bz})\right) \hat{\alpha} (t,\hat{\bz}),
\end{align}
which can also be directly recovered from the velocity Fourier transform of \eqref{eqn:mf:hamilton-vlasov}. \\

\noindent For the $N$ particle system, we compute from the Vlasov Hamiltonian the \textit{bosonic pseudo Schrödinger equation}
\begin{align}
\label{eqn:mf:pseudo-qm} \tag{PsQM\textsuperscript{N}}
\mi~ \partial_t \hat{\alpha}_N(t,\hat{\bz}_1,\dots,\hat{\bz}_N) = \left(\sum_{m=1}^{N} \nabla_{\bx_m} \cdot \nabla_{\xi_m} + \frac 1{2(N-1)}\sum_{n\neq m} \hat{V}(\hat{\bz}_n-\hat{\bz}_m)\right)~ \hat{\alpha}_N(t,\hat{\bz}_1,\dots,\hat{\bz}_N).
\end{align}
Its structure is similar to the bosonic many particle Schrödinger equation, the only difference being that the underlying space is $\RZHN$ instead of $\RXN$ and the kinetic term is different. In agreement with bosonic systems we will only be interested in symmetric states $\VF{\alpha}_N$, i.e., states that are invariant under coordinate permutations. 

\subsection{Relation of ensembles}

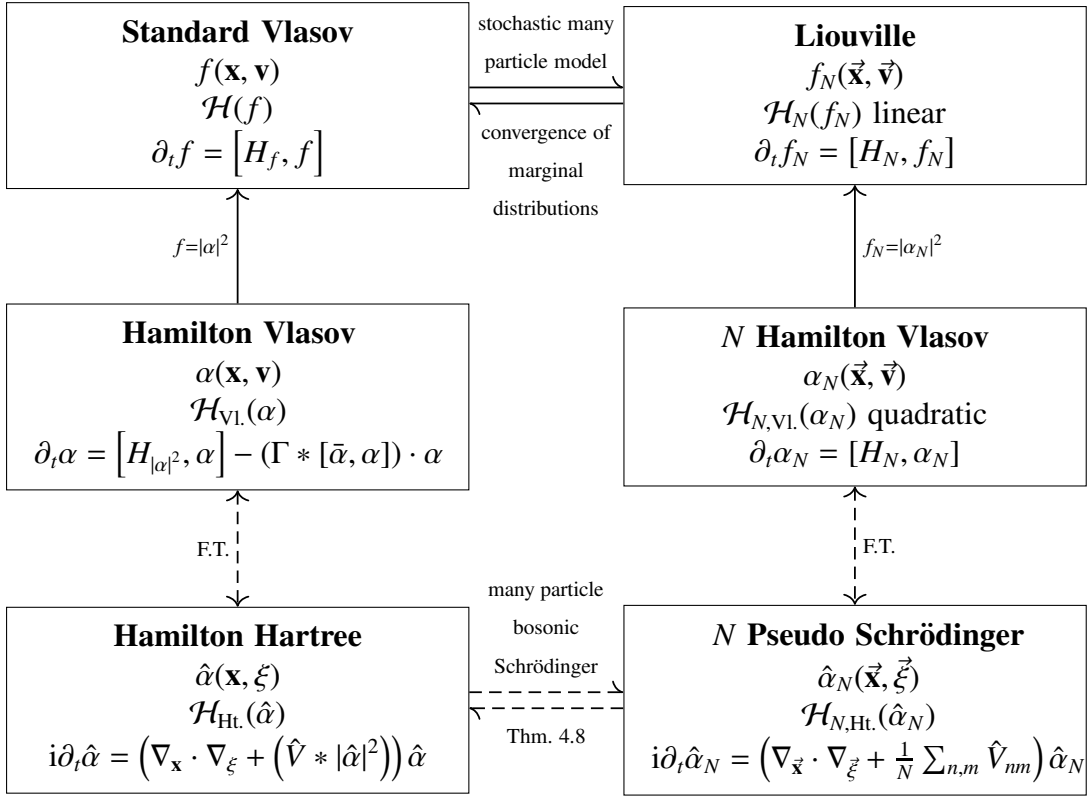
\begin{figure}[ht]
\centering
\begin{tikzcd}[ampersand replacement=\&, column sep=huge, row sep=huge, cells={nodes={draw=black, text width=0.38\textwidth, align=center}}]
\begin{array}{c} \text{\textbf{Standard Vlasov}} \\ f(\bx,\bv) \\ \HF(f) \\ \partial_t f = \symp{H_f}{f} \end{array}
\arrow[r, rightharpoonup, shift left, "\begin{array}{c}
\text{stochastic many} \\ \text{particle model}\end{array}"]
\& 
\begin{array}{c} \text{\textbf{Liouville}} \\ f_N(\vec{\bx},\vec{\bv}) \\ \HF_N(f_N)~ \text{linear} \\ \partial_t f_N = \symp{H_N}{f_N} \end{array}
\arrow[l,rightharpoonup, shift left, "\begin{array}{c}\text{convergence of} \\ \text{marginal} \\ \text{distributions}\end{array}"]
\\
\begin{array}{c} \text{\textbf{Hamilton Vlasov}} \\ \alpha(\bx,\bv) \\ \HVF(\alpha) \\ \partial_t \alpha = \symp{H_{\abs{\alpha}^2}}{\alpha} - \convol{\Gamma}{\symp{\bar{\alpha}}{\alpha}} \cdot \alpha \end{array} 
\arrow[d, dashed, leftrightarrow, swap, "\text{F.T.}"]
\arrow[u, rightarrow, "f=\abs{\alpha}^2"]
\&
\begin{array}{c} \text{\textbf{$N$ Hamilton Vlasov}} \\ \alpha_N(\vec{\bx},\vec{\bv}) \\ \HF_{N,\text{Vl.}}(\alpha_N)~\text{quadratic} \\ \partial_t \alpha_N = \symp{H_N}{\alpha_N} \end{array} 
\arrow[d, dashed, leftrightarrow, "\text{F.T.}"]
\arrow[u, rightarrow, swap, "f_N=\abs{\alpha_N}^2"]
\\
\begin{array}{c} \text{\textbf{Hamilton Hartree}} \\ \hat{\alpha}(\bx,\xi) \\ \HHF(\VF{\alpha}) \\ \mi\partial_t \hat{\alpha} = \left(\naX\cdot\naXi+\convol{\hat{V}}{\abs{\hat{\alpha}}^2}\right) \hat{\alpha} \end{array}
\arrow[r, dashed, rightharpoonup, shift left, "\begin{array}{c} \text{many particle} \\ \text{bosonic} \\ \text{Schrödinger} \end{array}"]
\&
\begin{array}{c} \text{\textbf{$N$ Pseudo Schrödinger}} \\ \hat{\alpha}_N(\vec{\bx},\vec{\xi}) \\ \HF_{N,\text{Ht.}}(\VF{\alpha}_N) \\ \mi\partial_t \hat{\alpha}_N = \left(\nabla_{\vec{\bx}} \cdot \nabla_{\vec{\xi}}+\frac 1N \sum_{n,m} \hat{V}_{nm}\right) \hat{\alpha}_N \end{array}
\arrow[l, dashed, rightharpoonup, shift left, "\begin{array}{c}\text{Thm. \ref{thm:mf:hamilton-hartree-mf-limit}}\end{array}"]
\end{tikzcd}
\caption{\label{fig:mf:eqn-hierarchy} Hierarchy of equations. All systems are given with name, phase space variable, energy functional, and equation of motion. $\symp{\cdot}{\cdot}$ denotes the Poisson bracket on $\RZ$, $\RZN$ respectively. All mean field theories are in the left column, all many particle systems in the right one. The velocity Fourier transform (F.T.) is reversible and therefore yields equivalent equations. The Vlasov mean field theory in the first row is widely discussed in the literature, e.g. \cite{braunhepp,lazarovicipickl}. The Hamiltonian Vlasov equation and its relation to Vlasov is discussed in \cite{neiss}. The dashed lines indicate new results established in this paper.}
\end{figure}

\noindent The connections between the various systems introduced above can be nicely represented in a \textit{commutative} diagram, as seen in Fig. \ref{fig:mf:eqn-hierarchy}. The diagram indicates two key features. 

At first, the structural analogy between the Hamilton Hartree/ $N$ Pseudo Schrödinger ensemble and the quantum mechanical Hartree/ Schrödinger system offers exciting opportunities to develop a mean field limit similar to the one of bosonic quantum systems. It is, with some restrictions and adaptations, possible to follow the guidelines of \cite{pickl}, where a quantum mean field limit of bosonic many-particle Schrödinger to the Hartree equation is proven. Following the strategy presented therein we will show in the present paper that $\hat{\alpha}(t)$ arises as a mean field limit from a sequence of many-particle pseudo wave functions $\hat{\alpha}_N(t)$ solving the \textit{$N$ pseudo Schrödinger system} \eqref{eqn:mf:pseudo-qm}. The sense of convergence is given by the average number of particles in the mean field state. This is rigorously carried out in Section \ref{sec:mf:mean-field-limit}.

Secondly, the diagram shows, that any mean field limit in the Hamiltonian Vlasov or Pseudo-Hartree picture yields a mean field limit for the classical Vlasov case by reversing the velocity Fourier transform and mapping $\alpha\mapsto \abs{\alpha}^2$. It is a peculiar example, where quantum $\Lp{2}{}$ methods can be used to infer properties of a classical system. Due to the reversibility of the velocity Fourier transform, $\Lp{2}{}$ mean field limits in both pictures are equivalent. \\

\section{Global well-posedness for regular potentials}

\label{sec:mf:regular-global-well-posedness}

As, to our knowledge, some of the $\Lp{2}{}$ systems are new, in order to get a valuable mean field discussion well-posedness of the underlying equations has to be established, first. Results for all four equations of motion \eqref{eqn:mf:hamilton-vlasov}, \eqref{eqn:mf:liouville-a}, \eqref{eqn:mf:hamilton-hartree}, and \eqref{eqn:mf:pseudo-qm}, can be obtained from solutions of \eqref{eqn:mf:hamilton-vlasov} and \eqref{eqn:mf:liouville-a}, because the Fourier transform immediately yields analogous results for \eqref{eqn:mf:hamilton-hartree} and \eqref{eqn:mf:pseudo-qm}. \\

\noindent While the algebraic structure of the problem requires to restrict the systems to non-relativistic, two body interaction, the technical restriction is that the physical force $\nabla\Gamma$ is $\Cp{2}{\bx}$ and bounded, i.e., the potential $\Gamma\in\Cm{3}{\RX;\RN}$, which is assumed to be even, satisfies
\begin{equation}
\label{eqn:mf:bounded-force} \tag{Pot}
C_\Gamma\equiv \max\left\{\Lpn{\Del{1}\Gamma}{\infty}{\bx}, \Lpn{\Del{2}\Gamma}{\infty}{\bx}, \Lpn{\Del{3}\Gamma}{\infty}{\bx}\right\} < \infty.
\end{equation}
By the regularity assumptions \eqref{eqn:mf:bounded-force} on the potential, global well-posedness results for the many-particle systems are not difficult to prove.

\begin{thm}[Global wellposed-ness for the $N$ particle system]
\label{thm:mf:existence-n-body}
Under the assumptions \eqref{eqn:mf:bounded-force}, the Liouville equation \eqref{eqn:mf:liouville-a} has a global solution for any $N$ and any initial value. The solution preserves regularity of the initial state up to $\Cp{2}{\vec{\bz}}$ as well as its symmetry under particle permutation. If the solution is classical in $\Cp{1}{(t,\bz)}$, it is unique therein.
Two solutions keep their $\Lp{2}{\vec{\bz}}$ distance constant for all times. \end{thm}

\begin{proof}
In fact, \eqref{eqn:mf:liouville-a} is a transport equation and its characteristics are independent of the initial value $\mathring{\alpha}_N$. They are  given by
\begin{align*}
\partial_t X_m(t,s,\bz) = V_m(t,s,\bz), \quad
\partial_t V_m(t,s,\bz) = \frac{1}{N-1} \sum_{n\neq m} \nabla\Gamma\left(X_n(t,s,\bz)-X_m(t,s,\bz)\right).
\end{align*}
Assumption \eqref{eqn:mf:bounded-force} assures that the solution map $\vec{Z}_N=(X_1,V_1,\dots,X_N,V_N)$ is well-defined on $\RN\times\RN\times\RZN$ and, because $\Gamma$ is $\Cp{3}{\bx}$, any transformation $\vec{Z}_N(t,s)\in\Cp{2}{\vec{\bz}}$. Hence, for any initial datum $\mathring{\alpha}_N\in\Lp{2}{\vec{\bz}}$ there is a solution given by $\alpha_N(t,\vec{\bz}) \equiv \mathring{\alpha}_N(\vec{Z}_N(0,t,\vec{\bz}))$. The solution is classical, if $\mathring{\alpha}_N\in\Cp{1}{\vec{\bz}}$ and it preserves regularity of the initial datum up to $\Cp{2}{\vec{\bz}}$ and also $\Wp{2,2}{\vec{\bz}}$, respectively.

In addition, the Hamiltonian structure of the characteristic equations implies that all $\Lp{p}{\vec{\bz}}$ norms are conserved. Ultimately, because the characteristic system is symmetric (under permutation of particles),  symmetric initial states are mapped onto symmetric states for all times. 

The uniqueness of classical solutions is immediate, because the transport equation and the characteristic equations are equivalent on that class.

For the last claim, we remark, that all solutions arise from composition with the very same solution map. It is volume preserving because of its Hamiltonian structure. Therefore one obtains   $\Lpn{\alpha_N(t)-\beta_N(t)}{2}{\vec{\bz}} = \Lpn{\mathring{\alpha}_N-\mathring{\beta}_N}{2}{\vec{\bz}}$ for any two initial data $\mathring{\alpha}_N,\mathring{\beta}_N$ and their respective solutions. \qedhere
\end{proof}

\noindent Proving a mean field limit requires bounds uniformly in $N$ on the first and second phase space derivatives of the solution. They can be naturally obtained by a respective condition on the sequence of $N$ particle Hamiltonians.

\begin{defn}[Mean field consistency]
\label{defn:mf:mf-consistency}
Let $\{H_N:\RZN\to\RN\}_{N\in\NN}$ be a sequence of Hamiltonian functions on the $N$ particle phase spaces. They are said to be \textbf{mean field consistent}, if and only if
\begin{enumerate}[label=(\roman*)]
\item each $H_N$ is invariant under particle permutation and
\item there are constants $C_1,C_2,C_3>0$, s.t. for any $N\geq 2$ and any permutation invariant $\alpha_N\in\Wp{2,2}{\vec{\bz}}$
\begin{align}
\nonumber
\Lpn{\symp{\Del{1}_{\bz_1} H_N}{\alpha_N}}{2}{\vec{\bz}} 
\leq&~ C_1 \Lpn{\Del{1}_{\bz_1}\alpha_N}{2}{\vec{\bz}}, \\ 
\label{eqn:mf:mean-field-consistency} \tag{MFC}
\Lpn{\symp{\Del{2}_{(\bz_1,\bz_2)} H_N}{\alpha_N}}{2}{\vec{\bz}} 
\leq&~ C_2 \Lpn{\Del{1}_{\bz_1}\alpha_N}{2}{\vec{\bz}}, \\
\nonumber
\Lpn{\symp{\Del{1}_{(\bz_1,\bz_2)} H_N}{\Del{1}_{(\bz_1,\bz_2)} \alpha_N}}{2}{\vec{\bz}}
\leq&~ C_3 \Lpn{\Del{2}_{(\bz_1,\bz_2)}\alpha_N}{2}{\vec{\bz}}.
\end{align}
\end{enumerate}
\end{defn}

\noindent This definition is of course tailored for our purposes.

\begin{lem}[Non-relativistic two body interaction]
\label{lem:mf:non-relativistic-two-body}
The sequence of Hamiltonians \eqref{eqn:mf:n-hamiltonian} along with condition \eqref{eqn:mf:bounded-force} is mean field consistent.
\end{lem}
\begin{proof}
For any $N\geq 2$ and $\alpha_N$ permutation invariant, w.l.o.g. $\Cp{\infty}{\vec{\bz}}$ and compactly supported, we have
\begin{align*}
&\Lpn{\symp{\Del{1}_{\bz_1}H_N}{\alpha_N}}{2}{\vec{\bz}} = \left(\int_{\RZN} \abs{\symp{\twovec{\frac{1}{N-1}\sum_{m=2}^{N} \nabla\Gamma(\bx_1-\bx_m)}{\bv_1}}{\alpha_N}}^2~ \intd\vec{\bz}\right)^{\frac{1}{2}} \\
\stackrel{(*)}{=}&~ \left(\int_{\RZN} \sum_{i=1}^{\ddim} \left(\abs{\partial_{x_{1,i}}\alpha_N}^2 + \abs{\sum_{j=1}^{\ddim} \frac{1}{N-1} \sum_{m=2}^{N} \partial_i\partial_j \Gamma(\bx_1-\bx_m)~ (\partial_{v_{1,j}}\alpha_N - \partial_{v_{m,j}}\alpha_N)}^2\right)~ \intd\vec{\bz}\right)^{\frac{1}{2}} \\
\stackrel{(**)}{\leq}&~ \sum_{i=1}^{\ddim} \left(\Lpn{\partial_{x_{1,i}}\alpha_N}{2}{\vec{\bz}} + 2\Lpn{\Del{2}\Gamma}{\infty}{\bx} \Lpn{\partial_{v_{1,i}}\alpha_N}{2}{\vec{\bz}}\right) \leq \sqrt{d\left(1+4\Lpn{\Del{2}\Gamma}{\infty}{\bx}^2\right)} \Lpn{\Del{1}_{\bz_1}\alpha_N}{2}{\vec{\bz}},
\end{align*}
where at $(*)$ we evaluate the Poisson bracket and norm square, at $(**)$ apply the matrix operator norm of $\Del{2}\Gamma$ and exploit the permutation invariance of $\alpha_N$ to use the derivative at the first particle index. 

The computations for the other two inequalities of \eqref{eqn:mf:mean-field-consistency} can be carried out very similarly. The bounds on the first three derivatives of $\Gamma$ are crucial here. By standard approximation, these inequalities are extended on permutation invariant functions of $\Wp{2,2}{\vec{\bz}}$. \qedhere
\end{proof}

\begin{lem}
\label{lem:mf:n-derivative-bounds}
Let $\alpha_N:\RN\times\RZN\to\CN$ a solution of \eqref{eqn:mf:liouville-a} invariant w.r.t. particle permutation, $\mathring{\alpha}_N=\alpha_N(0)\in\Wp{2,2}{\vec{\bz}}$, $\Lpn{\Del{1}_{\bz_1}\mathring{\alpha}_N}{2}{\vec{\bz}}, \Lpn{\Del{2}_{(\bz_1,\bz_2)}\mathring{\alpha}_N}{2}{\vec{\bz}} \leq M$. Let $H_N$ be mean field consistent, then we have for any $N\geq 2$ with $C_1,C_2,C_3>0$ from \eqref{eqn:mf:mean-field-consistency}
\begin{align}
\Lpn{\Del{1}_{\bz_1}\alpha_N(t)}{2}{\vec{\bz}} \leq&~ \Lpn{\Del{1}_{\bz_1}\mathring{\alpha}_N}{2}{\vec{\bz}}~ e^{\frac 12\left(1+C_1^2\right)t} \leq M e^{\frac 12\left(1+C_1^2\right)t} \equiv \bound{M,\Del{1}_{\bz_1}}{\Lp{2}{\vec{\bz}}}(t), \\
\Lpn{\Del{2}_{(\bz_1,\bz_2)}\alpha_N(t)}{2}{\vec{\bz}} \leq&~ \left(\Lpn{\Del{2}_{(\bz_1,\bz_2)}\mathring{\alpha}_N}{2}{\vec{\bz}}^2 + \frac{C_2^2}{1+C_1^2}\Lpn{\Del{1}_{\bz_1}\mathring{\alpha}_N}{2}{\vec{\bz}}^2 \left(e^{(1+C_1^2)t}-1\right)\right)^{\frac 12} e^{(\frac 32+C_3^2) t} \\
\leq&~ M \left(1+\frac{C_2^2}{1+C_1^2}\left(e^{\left(1+C_1^2\right)t}-1\right)\right)^{\frac 12} e^{\left(\frac{3}{2}+C_3^2\right)t} \equiv \bound{M,\Del{2}_{(\bz_1,\bz_2)}}{\Lp{2}{\vec{\bz}}}(t).
\end{align}
\end{lem}
\begin{proof}
Let $\alpha_N$ be the solution for some permutation invariant initial datum $\mathring{\alpha}_N$, which for now is supposed to be $\Cp{\infty}{\vec{\bz}}$ with compact support. Then for any $t\geq 0$
\begin{align*}
\Lpn{\Del{1}_{\bz_1}\alpha_N(t)}{2}{\vec{\bz}}^2 - \Lpn{\Del{1}_{\bz_1}\mathring{\alpha}_N}{2}{\vec{\bz}}^2 
=& \int_{0}^{t} \partial_\tau\Lpn{\Del{1}_{\bz_1}\alpha_N(\tau)}{2}{\vec{\bz}}^2 \intd\tau = 2\Re \int_{0}^{t} \SP{\Del{1}_{\bz_1}\alpha_N(\tau)}{\partial_\tau\Del{1}_{\bz_1}\alpha_N(\tau)} \intd\tau \\
=& -2 \Re\int_{0}^{t} \int_{\RZN} \left(\Del{1}_{\bz_1}\bar{\alpha}_N(\tau,\vec{\bz})~\cdot \Del{1}_{\bz_1} \symp{H_N}{\alpha_N(\tau)}(\vec{\bz})\right)~ \intd\vec{\bz}~ \intd\tau \\
=& -\int_{0}^{t} \underbrace{\int_{\RZN} \symp{H_N}{\abs{\Del{1}_{\bz_1}\alpha(\tau)}^2}(\vec{\bz})~ \intd\vec{\bz}}_{=0} \intd\tau \\
&-2 \Re\int_{0}^{t} \int_{\RZN} \left(\Del{1}_{\bz_1}\bar{\alpha}_N(\tau,\vec{\bz})~\cdot \symp{\Del{1}_{\bz_1}H_N}{\alpha_N(\tau)}(\vec{\bz})\right)~ \intd\vec{\bz}~ \intd\tau \\
\leq&~ 2 \int_{0}^{t} \Lpn{\Del{1}_{\bz_1}\alpha_N(\tau)}{2}{\vec{\bz}} \Lpn{\symp{\Del{1}_{\bz_1}H_N}{\alpha_N(\tau)}}{2}{\vec{\bz}}~ \intd\tau \\
\leq& \int_{0}^{t} \left(\Lpn{\Del{1}_{\bz_1}\alpha_N(\tau)}{2}{\vec{\bz}}^2 + \Lpn{\symp{\Del{1}_{\bz_1}H_N}{\alpha_N(\tau)}}{2}{\vec{\bz}}^2\right)~ \intd\tau \\
\leq&~ (1+C_1^2) \int_{0}^{t} \Lpn{\Del{1}_{\bz_1}\alpha_N(\tau)}{2}{\vec{\bz}}^2~ \intd\tau,
\end{align*}
yielding the first claim by Gronwall's Lemma. For the second claim, we have in a similar way
\begin{align*}
&\Lpn{\Del{2}_{(\bz_1,\bz_2)}\alpha_N(t)}{2}{\vec{\bz}}^2 - \Lpn{\Del{2}_{(\bz_1,\bz_2)}\mathring{\alpha}_N}{2}{\vec{\bz}}^2 
= \int_{0}^{t} \partial_\tau \Lpn{\Del{2}_{(\bz_1,\bz_2)}\alpha_N(\tau)}{2}{\vec{\bz}}^2~ \intd\tau \\
=&~ 2\Re \int_{0}^{t} \SP{\Del{2}_{(\bz_1,\bz_2)}\alpha_N(\tau)} {\Del{2}_{(\bz_1,\bz_2)}\partial_\tau\alpha_N(\tau)}~ \intd\tau 
= -\int_{0}^{t} \underbrace{\int_{\RZN} \symp{H_N}{\abs{\Del{2}_{(\bz_1,\bz_2)}\alpha_N(\tau)}^2}(\vec{\bz})~ \intd\vec{\bz}}_{=0}~ \intd\tau \\
&-2\Re \int_{0}^{t} \SP{\Del{2}_{(\bz_1,\bz_2)}\alpha_N(\tau)}{2\symp{\Del{1}_{(\bz_1,\bz_2)}H_N}{\Del{1}_{(\bz_1,\bz_2)}\alpha_N(\tau)} + \symp{\Del{2}_{(\bz_1,\bz_2)}H_N}{\alpha_N(\tau)}}~ \intd\tau \\
\leq&~ \int_{0}^{t} \left(3\Lpn{\Del{2}_{(\bz_1,\bz_2)}\alpha_N(\tau)}{2}{\vec{\bz}}^2 + 2\Lpn{\symp{\Del{1}_{(\bz_1,\bz_2)}H_N}{\Del{1}_{(\bz_1,\bz_2)}\alpha_N(\tau)}}{2}{\vec{\bz}}^2 + \Lpn{\symp{\Del{2}_{(\bz_1,\bz_2)}H_N}{\alpha_N(\tau)}}{2}{\vec{\bz}}^2\right)~ \intd\tau \\
\leq&~ C_2^2\int_{0}^{t} \Lpn{\Del{1}_{\bz_1}\alpha_N(\tau)}{2}{\vec{\bz}}^2~ \intd\tau + (3+2C_3^2)\int_{0}^{t} \Lpn{\Del{2}_{(\bz_1,\bz_2)}\alpha_N(\tau)}{2}{\vec{\bz}}^2~ \intd\tau.
\end{align*}
As intermediate steps in the computation require $\Cp{3}{\vec{\bz}}$ regularity of $\alpha_N$, it is necessary to extend this inequality in two steps. At first, consider a smoothened potential $\Gamma_\varepsilon\to\Gamma$ in the Hamiltonian $H_N$, s.t. the solution of the characteristic system is $\Cp{3}{\vec{\bz}}$ at least. As the support of $\alpha_N$ is compact in $I\times\RZN$ for any compact time interval $I\subset\RN$, both sides of this inequality converge. In a second step, the inequality can be extended to $\Wp{2,2}{\vec{\bz}}$, because the computation actually proves boundedness of the linear map
\begin{equation*}
\Wp{2,2}{\vec{\bz}} \to \Lp{\infty}{t,\text{loc}}\Lp{2}{\vec{\bz}}, \quad \mathring{\alpha}_N \mapsto \Del{2}_{(\bz_1,\bz_2)}\alpha_N. \qedhere
\end{equation*}
\end{proof}

\begin{thm}[Well-posedness for the mean field system]
\label{thm:mf:existence-hvl}
Let $\mathring{\alpha}\in\Wp{1,2}{\bz}$ be given. Then there is a unique global solution $\alpha(t,\bz)$ of \eqref{eqn:mf:hamilton-vlasov} with initial datum $\alpha(0)=\mathring{\alpha}$. For any $T<\infty$ the solution map
\begin{equation*}
\begin{array}{ccc}
\Wp{1,2}{\bz} & \to & \Lp{\infty}{t}\Lp{2}{\bz} \\
\mathring{\alpha} & \mapsto & \alpha : 
\begin{array}{ccc}
[0,T] & \to & \Lp{2}{\bz} \\
t & \mapsto & \alpha(t)
\end{array}
\end{array}
\end{equation*}
is locally Lipschitz continuous. Additionally, the image of this map is in $\Lp{\infty}{t}\Lp{2}{\bz}\cap\Lp{\infty}{t,\text{loc}}\Wp{1,2}{\bz}$ for $T=\infty$.
\end{thm}
\begin{proof}
\textbf{(i) Iteration scheme.} Consider $\mathring{\alpha}\in\Wp{1,2}{\bz}\cap\Cp{1}{\bz}$,  $\Wpn{\mathring{\alpha}}{1,2}{\bz}\leq M$, compactly supported. Choose any $T>0$. We define $\alpha_0(t,\bz)\equiv\mathring{\alpha}(\bz)$, and given $\alpha_n$, we define the objects
\begin{align*}
\rho_n(t,\bx) \equiv& \int_{\RV} \abs{\alpha_n(t,\bz)}^2~ \intd\bv, \\
F_n(t,\bx) \equiv& \convol{-\nabla\Gamma}{\rho_n(t)}(\bx), \\
K_n(t,\bx) \equiv& \convol{\Gamma}{\symp{\bar{\alpha}_n(t)}{\alpha_n(t)}}(\bx)
= \convol{\nabla\Gamma}{\int_{\RV} \bar{\alpha}_n(t,\bx,\bv)~ \naV\alpha_n(t,\bx,\bv)~ \intd\bv}(\bx),
\end{align*}
which are all well-defined for compactly supported $\alpha(t)\in\Cp{1}{\bz}$. In fact, $F_n(t),K_n(t)\in\Cp{2}{\bx}$ and all their derivatives are continuous in $(t,\bx)$. Further $Z_n(t,t_0,\bz)=(X_n,V_n)(t,t_0,\bz)$ is the solution map of the non-autonomous Hamiltonian system
\begin{equation}
\label{eqn:mf:iteration-characteristic-system}
\partial_t X_n(t,t_0,\bz) = V_n(t,t_0,\bz), \quad \partial_t V_n(t,t_0,\bz) = F_n(t,X_n(t,t_0,\bz)).
\end{equation}
Now $\alpha_{n+1}$ shall be defined iteratively from $\alpha_n$ as the unique solution of the linear equation
\begin{equation*}
\partial_t\alpha_{n+1}(t,\bz) = \symp{\frac{\abs{\bv}^2}{2} + \convol{\Gamma}{\abs{\alpha_n(t)}^2}(\bx)}{\alpha_{n+1}(t)}(\bz) - \convol{\Gamma}{\symp{\bar{\alpha}_n(t)}{\alpha_n(t)}}(\bx)~ \alpha_{n+1}(t,\bz)
\end{equation*}
with initial datum $\alpha_{n+1}(0)=\mathring{\alpha}$. In fact, we have the equation
\begin{equation}
\label{eqn:mf:iteration-transport}
\alpha_{n+1}(t,\bz) = \mathring{\alpha}(Z_n(0,t,\bz))~ \exp\left(-\int_{0}^{t} K_n(\tau,X_n(\tau,t,\bz))~ \intd\tau\right),
\end{equation}
which proves, that also $\alpha_{n+1}(t)$ will be compactly supported and $\Cp{1}{\bz}$ for any $t$. Therefore, the iteration scheme is well-defined. \\

\noindent\textbf{(ii) Uniform bounds.} In a next step, one determines bounds of the transport coefficients which are uniform in $n$ and depend on $M$, only. With $C_\Gamma$ from \eqref{eqn:mf:bounded-force},
\begin{align*}
\Lpn{\rho_n(t)}{1}{\bx} =& \Lpn{\alpha_n(t)}{2}{\bz}^2 \leq M^2, \\
\Lpn{F_n(t)}{\infty}{\bx} =& \Lpn{\convol{\nabla\Gamma}{\rho_n(t)}}{\infty}{\bx} \leq \Lpn{\nabla\Gamma}{\infty}{\bx} \Lpn{\rho_n(t)}{1}{\bx} \leq C_\Gamma M^2, \text{and} \\
\Lpn{\nabla F_n(t)}{\infty}{\bx} =& \Lpn{\convol{\Del{2}\Gamma}{\rho_n(t)}}{\infty}{\bx} \leq \Lpn{\Del{2}\Gamma}{\infty}{\bx} \Lpn{\rho_n(t)}{1}{\bx} \leq C_\Gamma M^2.
\end{align*}
This yields a bound on the solution map's differential in matrix operator norm
\begin{align*}
&&Z_n(t,t_0,\bz) =&~ \bz + \int_{t_0}^{t} \twovec{V_n(\tau,t_0,\bz)}{F_n(\tau,X_n(\tau,t_0,\bz))}~ \intd\tau \\
\Rightarrow&&\Del{}Z_n(t,t_0,\bz) =&~ \unity + \int_{t_0}^{t} \left(\begin{matrix}
0 & \unity \\ \nabla F_n(\tau,X_n(\tau,t_0,\bz)) & 0
\end{matrix}\right) \cdot \Del{}Z_n(\tau,t_0,\bz)~ \intd\tau \\
\Rightarrow&&~ \Lpn{\Del{}Z_n(t,t_0)}{\infty}{\bz} \leq&~ 1 + \int_{t_0}^{t} (1+C_\Gamma M^2) \Lpn{\Del{}Z_n(\tau,t_0)}{\infty}{\bz}~ \intd\tau \leq e^{(1+C_\Gamma M^2)\abs{t-t_0}}.
\end{align*}
Further one gets the following inequalities for $K_n$
\begin{align*}
\Lpn{K_n(t)}{\infty}{\bx} \leq&~ \Lpn{\Del{1}\Gamma}{\infty}{\bx} \Lpn{\alpha_n(t)}{2}{\bz} \Lpn{\nabla\alpha_n(t)}{2}{\bz} \leq C_\Gamma M \Lpn{\nabla\alpha_n(t)}{2}{\bz}~ \text{and} \\
\Lpn{\nabla K_n(t)}{\infty}{\bx} \leq&~ \Lpn{\Del{2}\Gamma}{\infty}{\bx} \Lpn{\alpha_n(t)}{2}{\bz} \Lpn{\nabla\alpha_n(t)}{2}{\bz} \leq C_\Gamma M \Lpn{\nabla\alpha_n(t)}{2}{\bz}.
\end{align*}
Combining all these results, one computes an iterative Gronwall scheme with help of \eqref{eqn:mf:iteration-transport}:
\begin{align*}
\Wpn{\alpha_{n+1}(t)}{1,2}{\bz} \leq& \Lpn{\alpha_{n+1}(t)}{2}{\bz} + \Lpn{\nabla\alpha_{n+1}(t)}{2}{\bz} \\
\leq& \Lpn{\mathring{\alpha}}{2}{\bz} + \Lpn{\nabla\mathring{\alpha}}{2}{\bz} \Lpn{\Del{}Z_n(0,t)}{\infty}{\bz} + \Lpn{\mathring{\alpha}}{2}{\bz} \int_{0}^{t} \Lpn{\nabla K_n(\tau)}{\infty}{\bx} \Lpn{\Del{}Z_n(\tau,t)}{\infty}{\bz}~ \intd\tau \\
\leq&~ M \left(1+e^{(1+C_\Gamma M^2) t}\right) + C_\Gamma M^2 \int_{0}^{t} \Wpn{\alpha_n(\tau)}{1,2}{\bz} e^{(1+C_\Gamma M^2)(t-\tau)}~ \intd\tau.
\end{align*}
Hence, if $Q_M: \RNp\to\RNp$ is the fixed point of the integral equation
\begin{equation*}
Q_M(t) = M \left(1+e^{(1+C_\Gamma M^2) t}\right) + C_\Gamma M^2 \int_{0}^{t} Q_M(\tau)~ e^{(1+C_\Gamma M^2)(t-\tau)}~ \intd\tau
\end{equation*}
given by
\begin{equation}
\label{eqn:mf:Qm-bound}
Q_M(t) = M\frac{1+C_\Gamma M^2+C_\Gamma M^2~ e^{(1+2C_\Gamma M^2)t}}{1+2C_\Gamma M^2},
\end{equation}
one obviously finds $\Wpn{\alpha_0(t)}{1,2}{\bz} \leq M \leq Q_M(t)$ and inductively $\Wpn{\alpha_n(t)}{1,2}{\bz} \leq Q_M(t)~\forall n$. \\

\noindent\textbf{(iii) Cauchy sequence.} The previous steps allow to compute an iterative Gronwall scheme in $\Lp{\infty}{t}\Lp{2}{\bz}$, $t\in[0,T]$:
\begin{align*}
&\Lpn{\alpha_{n+2}(t)-\alpha_{n+1}(t)}{2}{\bz}^2 = \int_{0}^{t} \partial_\tau\Lpn{\alpha_{n+2}(\tau)-\alpha_{n+1}(\tau)}{2}{\bz}^2~ \intd\tau \\
=&~ 2\Re\int_{0}^{t} \SP{\alpha_{n+2}(\tau)-\alpha_{n+1}(\tau)}{\partial_\tau\alpha_{n+2}(\tau)-\partial_\tau\alpha_{n+1}(\tau)}~\intd\tau \\
=&~ 2\Re\int_{0}^{t} \SP{\alpha_{n+2}(\tau)-\alpha_{n+1}(\tau)}{\symp{\frac{\abs{\bv}^2}{2}+\convol{\Gamma}{\abs{\alpha_{n+1}(\tau)}^2}}{\alpha_{n+2}(\tau)}} \intd\tau \\
&- 2\Re\int_{0}^{t} \SP{\alpha_{n+2}(\tau)-\alpha_{n+1}(\tau)}{\symp{\frac{\abs{\bv}^2}{2}+\convol{\Gamma}{\abs{\alpha_n(\tau)}^2}}{\alpha_{n+1}(\tau)}} \intd\tau \\
&- 2\Re\int_{0}^{t} \SP {\alpha_{n+2}(\tau)-\alpha_{n+1}(\tau)} {K_{n+1}(\tau)~ \alpha_{n+2}(\tau) - K_n(\tau)~ \alpha_{n+1}(\tau)}~ \intd\tau \\
=&~ 2\Re\int_{0}^{t} \SP{\alpha_{n+2}(\tau)-\alpha_{n+1}(\tau)}{\symp{\frac{\abs{\bv}^2}{2}+\convol{\Gamma}{\abs{\alpha_{n+1}(\tau)}^2}}{\alpha_{n+2}(\tau)-\alpha_{n+1}(\tau)}} \intd\tau \\
&+ 2\Re\int_{0}^{t} \SP{\alpha_{n+2}(\tau)-\alpha_{n+1}(\tau)}{\symp{\convol{\Gamma}{\left(\abs{\alpha_{n+1}(\tau)}^2-\abs{\alpha_n(\tau)}^2\right)}}{\alpha_{n+1}(\tau)}} \intd\tau \\
&- 2\Re\int_{0}^{t} \SP {\alpha_{n+2}(\tau)-\alpha_{n+1}(\tau)} {K_{n+1}(\tau)~ \left(\alpha_{n+2}(\tau) - \alpha_{n+1}(\tau)\right)}~ \intd\tau \\
&- 2\Re\int_{0}^{t} \SP {\alpha_{n+2}(\tau)-\alpha_{n+1}(\tau)} {\left(K_{n+1}(\tau) - K_n(\tau)\right)~ \alpha_{n+1}(\tau)}~ \intd\tau \\
\leq&~ \int_{0}^{t} \underbrace{\int_{\RZ} \symp{\frac{\abs{\bv}^2}{2}+\convol{\Gamma}{\abs{\alpha_{n+1}(\tau)}^2}}{\abs{\alpha_{n+2}(\tau)-\alpha_{n+1}(\tau)}^2}~ \intd\bz}_{=0}~ \intd\tau \\
&+ 2\int_{0}^{t} \Lpn{\alpha_{n+2}(\tau)-\alpha_{n+1}(\tau)}{2}{\bz} \Lpn{\symp{\convol{\Gamma}{\left(\abs{\alpha_{n+1}(\tau)}^2-\abs{\alpha_n(\tau)}^2\right)}}{\alpha_{n+1}(\tau)}}{2}{\bz} \intd\tau \\
&- 2\underbrace{\Re\int_{0}^{t} \int_{\RZ} \underbrace{K_{n+1}(\tau,\bx)}_{\in\mi\RN}~ \abs{\alpha_{n+2}(\tau,\bz) - \alpha_{n+1}(\tau,\bz)}^2~ \intd\bz~ \intd\tau}_{=0} \\
&+ 2\int_{0}^{t} \Lpn{\alpha_{n+2}(\tau)-\alpha_{n+1}(\tau)}{2}{\bz} \Lpn{\left(K_{n+1}(\tau) - K_n(\tau)\right)~ \alpha_{n+1}(\tau)}{2}{\bz}~ \intd\tau \\
\leq&~ 2\int_{0}^{t} \Lpn{\alpha_{n+2}(\tau)-\alpha_{n+1}(\tau)}{2}{\bz}^2~ \intd\tau + \int_{0}^{t} \Lpn{\symp{\convol{\Gamma}{\left(\abs{\alpha_{n+1}(\tau)}^2-\abs{\alpha_n(\tau)}^2\right)}}{\alpha_{n+1}(\tau)}}{2}{\bz}^2 \intd\tau \\
&+ \int_{0}^{t} \Lpn{\left(K_{n+1}(\tau) - K_n(\tau)\right)~ \alpha_{n+1}(\tau)}{2}{\bz}^2~ \intd\tau,
\end{align*}
where for the last two terms, we use that
\begin{align*}
&\Lpn{\symp{\convol{\Gamma}{\left(\abs{\alpha_{n+1}(\tau)}^2-\abs{\alpha_n(\tau)}^2\right)}}{\alpha_{n+1}(\tau)}}{2}{\bz} \\
\leq&~ \Lpn{\nabla\Gamma}{\infty}{\bx} \left(\Lpn{\alpha_{n+1}(\tau)}{2}{\bz} + \Lpn{\alpha_n(\tau)}{2}{\bz}\right) \Lpn{\alpha_{n+1}(\tau)-\alpha_n(\tau)}{2}{\bz} \Lpn{\nabla\alpha_{n+1}(\tau)}{2}{\bz} \\
\leq&~ 2C_\Gamma M Q_M(\tau) \Lpn{\alpha_{n+1}(\tau)-\alpha_n(\tau)}{2}{\bz} 
\end{align*}
and
\begin{align*}
&\Lpn{\left(K_{n+1}(\tau) - K_n(\tau)\right)~ \alpha_{n+1}(\tau)}{2}{\bz} 
\leq \Lpn{K_{n+1}(\tau)-K_n(\tau)}{\infty}{\bx} \Lpn{\alpha_{n+1}(\tau)}{2}{\bz} \\
\leq&~ \Lpn{\convol{\nabla\Gamma}{\int_{\RV} \left(\bar{\alpha}_{n+1}(\tau) \naV\alpha_{n+1}(\tau) - \bar{\alpha}_{n}(\tau) \naV\alpha_{n}(\tau)\right)~\intd\bv}}{\infty}{\bx} \Lpn{\alpha_{n+1}(\tau)}{2}{\bz} \\
\leq&~ \Lpn{\nabla\Gamma}{\infty}{\bx} \Lpn{\alpha_{n+1}(\tau)}{2}{\bz} \left(\Lpn{\nabla\alpha_{n+1}(\tau)}{2}{\bz} + \Lpn{\nabla\alpha_n(\tau)}{2}{\bz}\right) \Lpn{\alpha_{n+1}(\tau)-\alpha_n(\tau)}{2}{\bz} \\
\leq&~ 2C_\Gamma M Q_M(\tau) \Lpn{\alpha_{n+1}(\tau)-\alpha_n(\tau)}{2}{\bz}.
\end{align*}
Combining these estimates, along with an inductive Gronwall and $\Lpn{\alpha_1(\tau)-\alpha_0(\tau)}{2}{\bz}^2 \leq 4M^2$, one finds for any $0\leq t\leq T$
\begin{align*}
\Lpn{\alpha_{n+2}(t)-\alpha_{n+1}(t)}{2}{\bz}^2 \leq&~ 8C_\Gamma^2 M^2 \int_{0}^{t} Q_M^2(\tau) \Lpn{\alpha_{n+1}(\tau)-\alpha_{n}(\tau)}{2}{\bz}^2~ \intd\tau \\
&+ 2\int_{0}^{t} \Lpn{\alpha_{n+2}(\tau)-\alpha_{n+1}(\tau)}{2}{\bz}^2 \intd\tau \\
\leq&~ 8C_\Gamma^2 M^2 Q_M^2(t) e^{2t} \int_{0}^{t} \Lpn{\alpha_{n+1}(\tau)-\alpha_{n}(\tau)}{2}{\bz}^2~ \intd\tau \\
\leq&~ \left(8C_\Gamma^2 M^2 Q_M^2(T) e^{2T}\right)^{n+1} \frac{T^{n+1}}{(n+1)!} 4M^2.
\end{align*}
As the square root of the right-hand side is summable in $n$, $(\alpha_n)$ is a Cauchy sequence in the Banach space $\Lp{\infty}{}\left([0,T];\Lp{2}{\bz}\right)$ with a unique limit $\alpha\equiv\lim_n\alpha_n$. As $T$ was arbitrary, $\alpha$ is in fact in $\Lp{\infty}{t}\Lp{2}{\bz}$ for all $t\in\RNp$. \\

\noindent\textbf{(iv) Regularity.} Fix some $t\in\RNp$. (iii) proves, that $\alpha_n(t)\to\alpha(t)$ strongly in $\Lp{2}{\bz}$ and $\Wpn{\alpha_n(t)}{1,2}{\bz}\leq Q_M(t)$. By the reflexivity of the Hilbert space $\Wp{1,2}{\bz}$, some subsequence must converge weakly to the same limit $\alpha(t)$, hence $\alpha(t)\in\Wp{1,2}{\bz}$ and $\Wpn{\alpha(t)}{1,2}{\bz}\leq Q_M(t)$. 

In addition, we remark that by the strong $\Lp{2}{\bz}$ convergence, $F_n$ and $K_n$ as well as their first two spatial derivatives converge in $\Lp{\infty}{(t,\bx)}$, i.e., their limits are still continuous. Taking the $n\to\infty$ limit of \eqref{eqn:mf:iteration-characteristic-system} and \eqref{eqn:mf:iteration-transport} proves that $\alpha\in\Cp{1}{(t,\bz)}$ is a classical solution. \\

\noindent\textbf{(v) Dense extension and uniqueness.} Now assume, that $t\mapsto\alpha(t)$ and $t\mapsto\beta(t)$ are two classical solutions with compactly supported $\Cp{1}{\bz}$ initial values $\Wpn{\mathring{\alpha}}{1,2}{\bz},\Wpn{\mathring{\beta}}{1,2}{\bz}\leq M$. By using the bounds from (ii), which remain valid because solutions are fixed points of the iteration, and estimating a Gronwall type inequality similar to (iii), we find
\begin{align*}
&\Lpn{\alpha(t)-\beta(t)}{2}{\bz}^2 - \Lpn{\mathring{\alpha}-\mathring{\beta}}{2}{\bz}^2 - 2\int_{0}^{t} \Lpn{\alpha(\tau)-\beta(\tau)}{2}{\bz}^2~ \intd\tau \\
\leq& \int_{0}^{t} \Lpn{\symp{\convol{\Gamma}{\left(\abs{\alpha(\tau)}^2-\abs{\beta(\tau)}^2\right)}}{\alpha(\tau)}}{2}{\bz}^2 \intd\tau \\
&+ \int_{0}^{t} \Lpn{\convol{\nabla\Gamma}{\int_{\RV} \left(\bar{\alpha}(\tau) \naV\alpha(\tau) - \bar{\beta}(\tau) \naV\beta(\tau)\right)~ \intd\bv}~ \alpha(\tau)}{2}{\bz}^2~ \intd\tau \\
\leq& \int_{0}^{t} \Lpn{\nabla\Gamma}{\infty}{\bx}^2 \Lpn{\nabla\alpha(\tau)}{2}{\bz}^2 \left(\Lpn{\alpha(\tau)}{2}{\bz} + \Lpn{\beta(\tau)}{2}{\bz}\right)^2 \Lpn{\alpha(\tau)-\beta(\tau)}{2}{\bz}^2~ \intd\tau \\
&+ \int_{0}^{t} \Lpn{\nabla\Gamma}{\infty}{\bx}^2 \Lpn{\alpha(\tau)}{2}{\bz}^2 \left(\Lpn{\nabla\alpha(\tau)}{2}{\bz} + \Lpn{\nabla\beta(\tau)}{2}{\bz}\right)^2 \Lpn{\alpha(\tau)-\beta(\tau)}{2}{\bz}^2~ \intd\tau \\
\leq&~ 8C_\Gamma^2 M^2 \int_{0}^{t} Q_M^2(\tau) \Lpn{\alpha(\tau)-\beta(\tau)}{2}{\bz}^2~ \intd\tau,
\end{align*}
Gronwall's Lemma therefore yields
\begin{align}
\label{eqn:mf:hvl-lipschitz}
\Lpn{\alpha(t)-\beta(t)}{2}{\bz} \leq \Lpn{\mathring{\alpha}-\mathring{\beta}}{2}{\bz} \exp\left(\int_{0}^{t} \left(1+4C_\Gamma^2 M^2 Q_M^2(\tau)\right)~ \intd\tau\right).
\end{align}
Indeed, we have constructed a solution map 
\begin{equation*}
\Wp{1,2}{\bz}\cap\Cp{1}{\bz,c} \to \Lp{\infty}{t,\text{loc}}\Lp{2}{\bz}, \quad \mathring{\alpha} \mapsto \alpha,
\end{equation*}
which is locally Lipschitz and uniquely extensible to the entire space $\Wp{1,2}{\bz}$. This gives a notion of solution for every $\Wp{1,2}{\bz}$ initial value and \eqref{eqn:mf:hvl-lipschitz} also proves uniqueness of classical solutions with compact support. By the same regularity argument as in (iv), $\Wpn{\alpha(t)}{1,2}{\bz}\leq Q_{\Wpn{\mathring{\alpha}}{1,2}{\bz}}(t)$ remains valid for any initial datum. Hence, solutions are actually in $\Lp{\infty}{t,\text{loc}}\Wp{1,2}{\bz}\cap\Lp{\infty}{t}\Lp{2}{\bz}$. \qedhere
\end{proof}

\noindent Theorem \ref{thm:mf:existence-hvl} yields a couple of Corollaries. It states that solutions are bounded in $\Wp{1,2}{\bz}$ and explicitly gives this bound as well as the dependence of the Lipschitz constant of the solution map on the initial data.

\begin{cor}
\label{cor:mf:derivative-bound}
Let $\alpha: \RNp \to \Wp{1,2}{\bz}$ be a solution of \eqref{eqn:mf:hamilton-vlasov}. If $\Wpn{\alpha(0)}{1,2}{\bz}\leq M$, then the solution satisfies for any $t\geq0$
\begin{equation}
\label{eqn:mf:w12-bound}
\Wpn{\alpha(t)}{1,2}{\bz} \leq \bound{M}{\Wp{1,2}{\bz}}(t) = M\frac{1+C_\Gamma M^2+C_\Gamma M^2~ e^{(1+2C_\Gamma M^2)t}}{1+2C_\Gamma M^2}.
\end{equation}
\end{cor}

\begin{cor}
\label{cor:mf:lipschitz-dependence}
Let $M>0$ be arbitrary and $\mathring{\alpha}, \mathring{\beta}\in\Wp{1,2}{\bz}$ with $\Wpn{\mathring{\alpha}}{1,2}{\bz}, \Wpn{\mathring{\beta}}{1,2}{\bz}\leq M$. For the corresponding solutions $t\mapsto \alpha(t)$ and $t\mapsto \beta(t)$ of \eqref{eqn:mf:hamilton-vlasov}, we have
\begin{equation*}
\Lpn{\alpha(t)-\beta(t)}{2}{\bz} \leq \Lpn{\mathring{\alpha}-\mathring{\beta}}{2}{\bz} \exp\left(\int_{0}^{t} \left(1+4C_\Gamma^2 M^2 \left(\bound{M}{\Wp{1,2}{\bz}}(\tau)\right)^2\right)~ \intd\tau\right).
\end{equation*}
\end{cor}

\noindent Proving a mean field limit also requires bounds on the second derivative of solutions of \eqref{eqn:mf:hamilton-vlasov}. 

\begin{cor}
\label{cor:mf:2nd-derivative-bound}
There is a family of continuous, rising functions $\left\{\bound{M,\Del{2}}{\Lp{2}{\bz}}:\RNp\to\RNp\right\}_{M\geq0}$, s.t. for every solution $\alpha: \RNp \to \Wp{1,2}{\bz}$ of \eqref{eqn:mf:hamilton-vlasov} with initial value $\alpha(0)=\mathring{\alpha}\in\Wp{2,2}{\bz}$ and $\Wpn{\mathring{\alpha}}{2,2}{\bz}\leq M$ we have
\begin{align*}
\Lpn{\Del{2}\alpha(t)}{2}{\bz} \leq \bound{M,\Del{2}}{\Lp{2}{\bz}}(t)
\end{align*}
and the solution is infact in $\Lp{\infty}{t,\text{loc}}\Wp{2,2}{\bz}$.
\end{cor}

\begin{proof}
A similar claim for the $\Wp{1,2}{\bz}$ norm is already explicitly given in the proof of Theorem \ref{thm:mf:existence-hvl}. Hence, it suffices to consider $\Lpn{\Del{2}\alpha(t)}{2}{\bz}$. At first, we assume $\mathring{\alpha}\in\Cp{\infty}{\bz}$ compactly supported. It is necessary to bound the derivatives of the solution map $Z(t,t_0)$ of the characteristic system
\begin{align}
\nonumber
\Lpn{\Del{}Z(t,t_0)}{\infty}{\bz} \leq&~ 1+ \int_{t_0}^{t} \left(1 + \Lpn{\convol{\Del{2}\Gamma}{\abs{\alpha(\tau)}^2}}{\infty}{\bx}\right) \Lpn{\Del{}Z(\tau,t_0)}{\infty}{\bz}~ \intd\tau \\
\label{eqn:mf:solution-map-derivative-1-bound}
\leq&~ \exp\left(\left(1+C_\Gamma \Lpn{\mathring{\alpha}}{2}{\bz}^2\right)\abs{t-t_0}\right) \\
\nonumber
\Lpn{\Del{2}Z(t,t_0)}{\infty}{\bz} \leq& \int_{t_0}^{t} \Lpn{\convol{\Del{3}\Gamma}{\abs{\alpha(\tau)}^2}}{\infty}{\bx} \Lpn{\Del{}Z(\tau,t_0)}{\infty}{\bz}^2~ \intd\tau \\
\nonumber
&+ \int_{t_0}^{t} \left(1+\Lpn{\convol{\Del{2}\Gamma}{\abs{\alpha(\tau)}^2}}{\infty}{\bx}\right) \Lpn{\Del{2}Z(\tau,t_0)}{\infty}{\bz}~ \intd\tau \\
\nonumber
\leq&~ C_\Gamma \Lpn{\mathring{\alpha}}{2}{\bz}^2 \frac{\exp\left(2\left(1+C_\Gamma \Lpn{\mathring{\alpha}}{2}{\bz}^2\right)\abs{t-t_0}\right) - 1}{2\left(1+C_\Gamma \Lpn{\mathring{\alpha}}{2}{\bz}^2\right)} \exp\left(\left(C_\Gamma \Lpn{\mathring{\alpha}}{2}{\bz}^2\right)\abs{t-t_0}\right) \\
\label{eqn:mf:solution-map-derivative-2-bound}
\leq&~ \frac{1}{2} \exp\left(2\left(1+C_\Gamma \Lpn{\mathring{\alpha}}{2}{\bz}^2\right)\abs{t-t_0}\right).
\end{align}
For the second derivative of $\alpha(t)$, we denote $K(t,\bx)\equiv\convol{\Gamma}{\symp{\bar{\alpha}(t)}{\alpha(t)}}(\bx)$ and recall the transport formula for general Hamiltonian Vlasov systems \cite[Prop.2.2]{neiss}
\begin{equation*}
\alpha(t,\bz) = \mathring{\alpha}(Z(0,t,\bz)) \exp\left(-\int_{0}^{t} K(\tau,X(\tau,t,\bz))~ \intd\tau\right).
\end{equation*}
With this at hand, we determine
\begin{align*}
&\left(\partial_l\partial_k \alpha(t,\bz)\right)~ \exp\left(\int_{0}^{t} K(\tau,X(\tau,t,\bz))~ \intd\tau\right) \\
=& \sum_{i,j} \partial_i\partial_j\mathring{\alpha}(Z(0,t,\bz))~ \partial_l Z_i(0,t,\bz)~ \partial_k Z_j(0,t,\bz) \\
&+ \sum_j \partial_j\mathring{\alpha}(Z(0,t,\bz))~ \partial_k\partial_l Z_j(0,t,\bz) \\
&- \sum_i\partial_i\mathring{\alpha}(Z(0,t,\bz))~ \partial_l Z_i(0,t,\bz)~ \int_{0}^{t} \sum_j \partial_j K(\tau,X(\tau,t,\bz))~ \partial_k X_j(\tau,t,\bz)~ \intd\tau \\
&- \sum_i\partial_i\mathring{\alpha}(Z(0,t,\bz))~ \partial_k Z_i(0,t,\bz)~ \int_{0}^{t} \sum_j \partial_j K(\tau,X(\tau,t,\bz))~ \partial_l X_j(\tau,t,\bz)~ \intd\tau \\
&- \mathring{\alpha}(Z(0,t,\bz))~ \int_{0}^{t} \sum_{i,j} \partial_i\partial_j K(\tau,X(\tau,t,\bz))~ \partial_k X_j(\tau,t,\bz)~ \partial_l X_i(\tau,t,\bz)~ \intd\tau \\
&- \mathring{\alpha}(Z(0,t,\bz))~ \int_{0}^{t} \sum_j \partial_j K(\tau,X(\tau,t,\bz))~ \partial_l\partial_k X_j(\tau,t,\bz)~ \intd\tau,
\end{align*}
hence the $\Lp{2}{\bz}$ norm is computed to be
\begin{align}
\nonumber
\Lpn{\Del{2}\alpha(t)}{2}{\bz} \leq&~ \Lpn{\Del{2}\mathring{\alpha}}{2}{\bz} \Lpn{\Del{}Z(0,t)}{\infty}{\bz}^2 + \sqrt{2\ddim} \Lpn{\nabla\mathring{\alpha}}{2}{\bz} \Lpn{\Del{2}Z(0,t)}{\infty}{\bz} \\
\nonumber
&+ 2\sqrt{2\ddim} \Lpn{\nabla\mathring{\alpha}}{2}{\bz} \Lpn{\Del{}Z(0,t)}{\infty}{\bz} \int_{0}^{t} \Lpn{\convol{\nabla\Gamma}{\symp{\bar{\alpha}(\tau)}{\alpha(\tau)}}}{\infty}{\bx} \Lpn{\Del{}Z(\tau,t)}{\infty}{\bz}~ \intd\tau \\
\nonumber
&+ \Lpn{\mathring{\alpha}}{2}{\bz} \int_{0}^{t} \Lpn{\convol{\Del{2}\Gamma}{\symp{\bar{\alpha}(\tau)}{\alpha(\tau)}}}{\infty}{\bx} \Lpn{\Del{}Z(\tau,t)}{\infty}{\bz}^2~ \intd\tau \\
\label{eqn:mf:d2-bound}
&+ \Lpn{\mathring{\alpha}}{2}{\bz} \int_{0}^{t} \Lpn{\convol{\nabla\Gamma}{\symp{\bar{\alpha}(\tau)}{\alpha(\tau)}}}{\infty}{\bx} \Lpn{\Del{2}Z(\tau,t)}{\infty}{\bz}~ \intd\tau \leq \bound{M,\Del{2}}{\Lp{2}{\bz}}(t)
\end{align}
for a continuous rising function $\bound{M,\Del{2}}{\Lp{2}{\bz}}:\RNp\to\RNp$ explicitly computable from \eqref{eqn:mf:solution-map-derivative-1-bound}, \eqref{eqn:mf:solution-map-derivative-2-bound}, and $\bound{M}{\Wp{1,2}{\bz}}$. In particular, it only depends on a bound $M$ of the $\Wp{2,2}{\bz}$ norm of $\mathring\alpha$. Now if we approximate $\mathring{\alpha}\in\Wp{2,2}{\bz}$ by a sequence $\mathring{\alpha}^{(k)}$ of test functions with norm $\leq M$, we have norm convergence $\Lpn{\alpha(t)-\alpha^{(k)}(t)}{2}{\bz}\to 0$ for any $t\geq 0$ by Theorem \ref{thm:mf:existence-hvl}. By the inequalities \eqref{eqn:mf:w12-bound}, \eqref{eqn:mf:d2-bound} given above, $\Wpn{\alpha^{(k)}(t)}{2,2}{\bz}$ is bounded and the reflexivity of $\Wp{2,2}{\bz}$ gives a weakly converging subsequence. Its limit must coincide with the $\Lp{2}{\bz}$ limit $\alpha(t)$ which therefore is already in $\Wp{2,2}{\bz}$ and its norm is dominated by the same bound \eqref{eqn:mf:d2-bound}, i.e., $\Lpn{\Del{2}\alpha(t)}{2}{\bz}\leq \bound{M,\Del{2}}{\Lp{2}{\bz}}(t)$. \qedhere
\end{proof}

\begin{defn}
For any $k\in\NN_0$, define $\Mp{k}{\VF{\bz}}\equiv\fourier\left(\Wp{k,2}{\bz}\right)$ the set of functions $\VF{\alpha}: \RZH \to \CN$ that arise from velocity Fourier transform \eqref{eqn:mf:velocity-fourier} with isometric norm $\Mpn{\VF{\alpha}}{k}{\VF{\alpha}}\equiv\Wpn{\inv{\fourier}\VF{\alpha}}{k,2}{\bz}$. In the case of $k=1$ the norm is, e.g.,
\begin{equation*}
\Mpn{\VF{\alpha}}{1}{\VF{\bz}} \equiv \left(\Lpn{\VF{\alpha}}{2}{\VF{\bz}}^2 + \Lpn{\abs{\xi}\VF{\alpha}}{2}{\VF{\bz}}^2 + \Lpn{\naX\VF{\alpha}}{2}{\VF{\bz}}^2\right)^{\frac{1}{2}}.
\end{equation*}
Likewise, one defines $\Mp{k}{\vec{\VF{\bz}}}\equiv \fourier^{\otimes N} \left(\Wp{k,2}{\vec{\bz}}\right)$ for the $N$ particle states.
\end{defn}

\begin{thm}[Well-posedness for the Hamilton Hartree system]
\label{thm:mf:existence-hht}
Let $\mathring{\VF{\alpha}}\in\Mp{1}{\VF{\bz}}$ be arbitrary. Then there is a unique global solution $\VF{\alpha}(t,\VF{\bz})$ of \eqref{eqn:mf:hamilton-hartree} with initial datum $\VF{\alpha}(0)=\mathring{\VF{\alpha}}$. For any $T<\infty$ the solution map
\begin{equation*}
\begin{array}{ccc}
\Mp{1}{\VF{\bz}} & \to & \Lp{\infty}{t}\Lp{2}{\VF{\bz}} \\
\mathring{\VF{\alpha}} & \mapsto & \VF{\alpha} :
\begin{array}{ccc}
[0,T] & \to & \Lp{2}{\VF{\bz}} \\
t & \mapsto & \VF{\alpha}(t)
\end{array}
\end{array}
\end{equation*}
is locally Lipschitz continuous. In fact, the image of this map is in $\Lp{\infty}{t}\Lp{2}{\VF{\bz}}\cap\Lp{\infty}{t,\text{loc}}\Mp{1}{\VF{\bz}}$ for $T=\infty$.
\end{thm}
\begin{proof}
This is immediate from the norm conserving properties of the Fourier transform, $\Wpn{\alpha}{1,2}{\bz}=\Mpn{\VF{\alpha}}{1}{\VF{\bz}}$ and $\Lpn{\alpha}{2}{\bz}=\Lpn{\VF{\alpha}}{2}{\VF{\bz}}$, 
and Thm. \ref{thm:mf:existence-hvl}. We also remark that classical compactly supported solutions of \eqref{eqn:mf:hamilton-vlasov} will transform into classical solutions of \eqref{eqn:mf:hamilton-hartree}. Therefore, the notion of a solution is meaningful, even without proper regularity restrictions for the derivatives. \qedhere
\end{proof}

\section{Mean field limit}

\label{sec:mf:mean-field-limit}

\newcommand{\cn}{\mathfrak{n}}
\newcommand{\cm}{\mathfrak{m}}
\newcommand{\cml}{\cm_\lambda}
\newcommand{\dCml}[1]{\Delta_{{#1}}\cml}
\newcommand{\shift}[2]{\left({#1}\right)_{#2}}
\newcommand{\ExN}[2]{\SP{\alpha_N({#2})} {{#1} \alpha_N({#2})}}
\newcommand{\ExHN}[2]{\SP{\VF{\alpha}_N({#2})} {{#1} \VF{\alpha}_N({#2})}}
\newcommand{\proj}[2]{p_{{#1}}^{{#2}}}
\newcommand{\qroj}[2]{q_{{#1}}^{{#2}}}
\newcommand{\Proj}[3]{P_{{#1},{#2}}^{{#3}}}
\newcommand{\SFS}{\HS}
\newcommand{\MFS}[1]{\SFS_{{#1}}}

In this chapter, let $\SFS\equiv\Lp{2}{\VF{\bz}}$ denote the phase space for the effective one-particle description. For any $N\in\NN$ define $\MFS{N}\equiv\SFS^{\otimes N} \subseteq \Lp{2}{(\VF{\bz}_1,\dots,\VF{\bz}_N)}\equiv\Lp{2}{\vec{\VF{\bz}}}$. It contains the bosonic pseudo Fock space of $N$ particle states, which is given by all functions that are invariant under particle index permutation, referred to also as \textit{symmetric}. \\

\noindent The mean field derivation heavily follows ideas and uses techniques of \cite{pickl} with slight adjustments to fit the new requirements of an unbounded, non-integrable pair interaction potential. For the sake of completeness, the key definitions are included here. The main result is Theorem \ref{thm:mf:hamilton-hartree-mf-limit} at the end of the section.

\begin{defn}
[$N$ particle projections and counting operators]
\cite[Def.2.1]{pickl}
Fix some normed $\VF{\alpha}\in\SFS$, $\norm{\VF{\alpha}}_\SFS=1$ and some $N\in\NN$.
\begin{enumerate}[label=(\roman*)]
\item For $j\in\{1,\dots,N\}$ and $k\in\IN$ define the projections in $\MFS{N}$
\begin{align*}
\left(\proj{j}{\VF{\alpha}} \VF{\alpha}_N\right)(\vec{\VF{\bz}}) \equiv&~ \VF{\alpha}(\VF{\bz}_j) \int_{\RZH} \bar{\VF{\alpha}}(\tilde{\VF{\bz}}_j)~ \VF{\alpha}_N(\VF{\bz}_1,\dots,\tilde{\VF{\bz}}_j,\dots,\VF{\bz}_N)~ \intd\tilde{\VF{\bz}}_j, \\
\qroj{j}{\VF{\alpha}} \equiv&~ \id_{\MFS{N}}-\proj{j}{\VF{\alpha}}, \\
\Proj{k}{j}{\VF{\alpha}} \equiv&~ \sum_{\substack{\ba\in\{0,1\}^j \\ \sum_i a_i = k}} \prod_{i=1}^j \left(\proj{N-j+i}{\VF{\alpha}}\right)^{1-a_i}~ \left(\qroj{N-j+i}{\VF{\alpha}}\right)^{a_i},
\end{align*}
which all commute pairwise.
\item Now let $f: \{0,\dots,N\}\to\RN$ be a map. We define the associated operator on $\MFS{N}$ by
\begin{equation*}
f^{\VF{\alpha}} \equiv \sum_{k=0}^{N} f(k)~ \Proj{k}{N}{\VF{\alpha}},
\end{equation*}
i.e., $f$ defines the spectral decomposition of $f^{\VF{\alpha}}$ and the operator commutes with all projections from (i). In addition, for $l\in\IN$, define the shifted operator
\begin{equation*}
\shift{f^{\VF{\alpha}}}{l} \equiv \sum_{k=l}^{N+l} f(k-l)~ \Proj{k}{N}{\VF{\alpha}} = \sum_{k=0}^{N} f(k) \Proj{k+l}{N}{\VF{\alpha}}.
\end{equation*}
\item In particular, we define for any $\lambda\in[0,1]$ the \textbf{counting functions} on $\{0,\dots,N\}$
\begin{align*}
\cn(k) \equiv \sqrt{\frac{k}{N}}, \quad \text{and} \quad \cml(k) \equiv \min\left\{\frac{k}{N^\lambda},1\right\}.
\end{align*}
\end{enumerate}
\end{defn}

\begin{rmk}
\textbf{(i).} $\proj{j}{\VF{\alpha}}$ projects onto the subspace of functions with $j$-th particle in state $\VF{\alpha}$. $\Proj{k}{j}{\VF{\alpha}}$ projects onto the subspace, where exactly $j-k$ out of the last $j$ particles are in state $\VF{\alpha}$ and the other $k$ particles are in an orthogonal state. \\

\noindent\textbf{(ii).} Note, that $\Proj{k}{j}{\VF{\alpha}}\neq 0$ only for $0\leq k\leq j\leq N$. In addition, $\sum_{k=0}^j \Proj{k}{j}{\VF{\alpha}} = \id_{\MFS{N}}$. \\

\noindent\textbf{(iii).} $\norm{\cn^{\VF{\alpha}}\VF{\alpha}_N}^2 = \SP{\VF{\alpha}_N}{(\cn^{\VF{\alpha}})^2\VF{\alpha}_N}$ counts the share of particles in $\VF{\alpha}_N$ not in state $\VF{\alpha}$. Also, because $\cn(k)^2\leq \cml(k)$, the quantity $\beta_N\equiv \SP{\VF{\alpha}_N}{\cml^{\VF{\alpha}}\VF{\alpha}_N}$ controls the distance of $\VF{\alpha}_N$ and the pure product state $\VF{\alpha}^{\otimes N}$ in a suitable way.
\end{rmk}

\begin{lem}
\label{lem:mf:computation-rules}
\cite[Lem.2.3]{pickl}
Let $N\in\NN$ be fixed and $\VF{\alpha}\in\Lp{2}{\VF{\bz}}$ be normalized.
\begin{enumerate}[label=(\roman*)]
\item For any two $f,g: \{0,\dots,N\}\to\RNp$ we have
\begin{equation*}
f^{\VF{\alpha}}~ g^{\VF{\alpha}} = g^{\VF{\alpha}}~ f^{\VF{\alpha}}, \quad 
f^{\VF{\alpha}} \proj{j}{\VF{\alpha}} = \proj{j}{\VF{\alpha}} f^{\VF{\alpha}}, \quad 
f^{\VF{\alpha}} \Proj{k}{j}{\VF{\alpha}} = \Proj{k}{j}{\VF{\alpha}} f^{\VF{\alpha}} \quad \forall j.
\end{equation*}
\item The following  identities hold
\begin{equation*}
\frac{1}{N} \sum_{j=1}^{N} \qroj{j}{\VF{\alpha}} = \frac{1}{N} \sum_{j=1}^{N} \qroj{j}{\VF{\alpha}} \sum_{k=0}^N \Proj{k}{N}{\VF{\alpha}} = \sum_{k=0}^{N} \frac{k}{N}~ \Proj{k}{N}{\VF{\alpha}} = \left(\cn^{\VF{\alpha}}\right)^2.
\end{equation*}
\item For any $f: \{0,\dots,N\}\to\RNp$ and $\alpha_N\in\Lp{2}{\vec{\VF{\bz}}}$ symmetric we have
\begin{equation*}
\norm{f^{\VF{\alpha}} \qroj{1}{\VF{\alpha}} \VF{\alpha}_N} = \norm{f^{\VF{\alpha}} \cn^{\VF{\alpha}} \VF{\alpha}_N}, \quad \norm{f^{\VF{\alpha}} \qroj{1}{\VF{\alpha}} \qroj{2}{\VF{\alpha}} \VF{\alpha}_N} \leq \sqrt{\frac{N}{N-1}} \norm{f^{\VF{\alpha}} \left(\cn^{\VF{\alpha}}\right)^2 \VF{\alpha}_N}.
\end{equation*}
\item For any $f: \{0,\dots,N\} \to \RNp$, $\VF{v}:\left(\RZH\right)^2\to\RN$ multiplication operator, and any $j,k\in\{0,1,2\}$ we have
\begin{equation*}
f^{\VF{\alpha}} Q_j^{\VF{\alpha}}~ \VF{v}(\VF{\bz}_1,\VF{\bz}_2)~ Q_k^{\VF{\alpha}} = Q_j^{\VF{\alpha}}~ \VF{v}(\VF{\bz}_1,\VF{\bz}_2)~ Q_k^{\VF{\alpha}} \shift{f^{\VF{\alpha}}}{k-j},
\end{equation*}
where $Q_0^{\VF{\alpha}}=\proj{1}{\VF{\alpha}} \proj{2}{\VF{\alpha}}$, $Q_1^{\VF{\alpha}}=\proj{1}{\VF{\alpha}} \qroj{2}{\VF{\alpha}}$, and $Q_2^{\VF{\alpha}}=\qroj{1}{\VF{\alpha}} \qroj{2}{\VF{\alpha}}$.
\end{enumerate}
\end{lem}
\begin{proof}
See reference. \qedhere
\end{proof}

\begin{lem}
\label{lem:mf:derivative}
Let $\VF{\alpha}_N: \RNp \to \Mp{1}{\vec{\VF{\bz}}}$ be some $\Lp{2}{\vec{\VF{\bz}}}$ normalized symmetric solution of \eqref{eqn:mf:pseudo-qm}. Let $\VF{\alpha}: \RNp\to\Mp{1}{\VF{\bz}}$ be some $\Lp{2}{\bz}$ normalized solution of \eqref{eqn:mf:hamilton-hartree}. Then we have
\begin{align*}
&\ExHN{\cml^{\VF{\alpha}(t)}}{t}-\ExHN{\cml^{\VF{\alpha}(0)}}{0} \\
=&~ N \int_{0}^{t} \Im \ExHN{\cml^{\VF{\alpha}(\tau)} \left(\VF{V}(\VF{\bz}_1-\VF{\bz}_2) - \convol{\VF{V}}{\abs{\VF{\alpha}(\tau)}^2}(\VF{\bz}_1) - \convol{\VF{V}}{\abs{\VF{\alpha}(\tau)}^2}(\VF{\bz}_2)\right)}{\tau}~ \intd\tau.
\end{align*}
\end{lem}

\begin{proof}
At first we assume that $\VF{\alpha}(0)$ and $\VF{\alpha}_N(0)$ are Schwartz functions. Their smoothness and decay properties are then uniform on a compact time intervals. The time derivative then can be pulled into the inner product and be explicitly evaluated. Equations \eqref{eqn:mf:hamilton-hartree} and \eqref{eqn:mf:pseudo-qm} immediately apply, i.e.,
\begin{align*}
&\ExHN{\cml^{\VF{\alpha}(t)}}{t}-\ExHN{\cml^{\VF{\alpha}(0)}}{0} 
= \int_{0}^{t} \partial_\tau \ExHN{\cml^{\VF{\alpha}(\tau)}}{\tau}~ \intd\tau \\
=& \int_{0}^{t} \left(2\Re \ExHN{\cml^{\VF{\alpha}(\tau)} \partial_\tau}{\tau} + \ExHN{\left(\partial_\tau\cml^{\VF{\alpha}(\tau)}\right)}{\tau}\right) \intd\tau = (*).
\end{align*}
For Schwartz functions $\partial_\tau\VF{\alpha}(\tau) = \frac{1}{\mi}\VF{H}^{\VF{\alpha}(\tau)} \VF{\alpha}(\tau)$, one computes
\begin{equation*}
\partial_\tau \proj{m}{\VF{\alpha}(\tau)} = \frac{1}{\mi} \left(\VF{H}_m^{\VF{\alpha}(\tau)} \proj{m}{\VF{\alpha}(\tau)} - \proj{m}{\VF{\alpha}(\tau)} \VF{H}_m^{\VF{\alpha}(\tau)}\right), \quad \partial_\tau \Proj{k}{N}{\VF{\alpha}(\tau)} = \frac{1}{\mi} \sum_{m=1}^{N} \left(\VF{H}_m^{\VF{\alpha}(\tau)} \Proj{k}{N}{\VF{\alpha}(\tau)} - \Proj{k}{N}{\VF{\alpha}(\tau)} \VF{H}_m^{\VF{\alpha}(\tau)}\right),
\end{equation*}
yielding along with $\partial_\tau\VF{\alpha}_N(\tau) = \frac{1}{\mi}\VF{H}_N \VF{\alpha}_N(\tau)$ and the symmetry of $\VF{H}_N$ and $\VF{H}_m^{\VF{\alpha}}$
\begin{align*}
(*) =&~ 2\int_{0}^{t} \Im \ExHN{\cml^{\VF{\alpha}(\tau)} \left(\VF{H}_N - \sum_{m} \VF{H}_m^{\VF{\alpha}(\tau)}\right)}{\tau} \intd\tau \\
=&~ 2\int_{0}^{t} \Im \ExHN{\cml^{\VF{\alpha}(\tau)} \left(\frac{1}{2(N-1)} \sum_{m\neq n} \VF{V}(\VF{\bz}_m-\VF{\bz}_n) - \sum_{m} \convol{\VF{V}}{\abs{\VF{\alpha}(\tau)}^2}(\VF{\bz}_m)\right)}{\tau} \intd\tau \\
=&~ N\int_{0}^{t} \Im \ExHN{\cml^{\VF{\alpha}(\tau)} \left(\VF{V}(\VF{\bz}_1-\VF{\bz}_2) - \convol{\VF{V}}{\abs{\VF{\alpha}(\tau)}^2}(\VF{\bz}_1) - \convol{\VF{V}}{\abs{\VF{\alpha}(\tau)}^2}(\VF{\bz}_2)\right)}{\tau} \intd\tau,
\end{align*}
where we made use of the symmetry of $\alpha_N$ at the end. The claim can  now be completed by standard approximation on $\Mp{1}{\VF{\bz}}$ for initial data and implied $\Lp{2}{\VF{\bz}}$ convergence of their respective solutions on compact time intervals. \qedhere
\end{proof}

\begin{lem}
\cite[Sec.3.2]{pickl}
\label{lem:mf:m-decomposition}
Define the counting functions on $\{0,\dots,N\}$
\begin{equation*}
\dCml{l}(k) \equiv \left\{\begin{array}{ccc}
\cml(k)-\cml(k-l) & : & -\abs{l}\leq k \leq N^\lambda+\abs{l}, \\
0 & : & \text{else}.
\end{array}\right.
\end{equation*}
For any normed $\VF{\alpha}\in\Lp{2}{\VF{\bz}}$ and fixed $N\geq 2$, we have
\begin{equation*}
\cml^{\VF{\alpha}} = \left(\dCml{-2}^{\VF{\alpha}}\right) \proj{1}{\VF{\alpha}} \proj{2}{\VF{\alpha}} + \left(\dCml{-1}^{\VF{\alpha}}\right) \left(\proj{1}{\VF{\alpha}} \qroj{2}{\VF{\alpha}} + \qroj{1}{\VF{\alpha}} \proj{2}{\VF{\alpha}}\right) + \sum_{k=0}^{N} \cml(k) \Proj{k-2}{N-2}{\VF{\alpha}}.
\end{equation*}
\end{lem}

\begin{proof}
\begin{align*}
\cml^{\VF{\alpha}} =& \sum_{k=0}^{N} \cml(k) \left(\proj{1}{\VF{\alpha}} \proj{2}{\VF{\alpha}} \Proj{k}{N-2}{\VF{\alpha}} + \left(\proj{1}{\VF{\alpha}} \qroj{2}{\VF{\alpha}} + \qroj{1}{\VF{\alpha}} \proj{2}{\VF{\alpha}}\right) \Proj{k-1}{N-2}{\VF{\alpha}} + \qroj{1}{\VF{\alpha}} \qroj{2}{\VF{\alpha}} \Proj{k-2}{N-2}{\VF{\alpha}}\right) \\
=& \sum_{k=0}^{N} \cml(k) \left(\Proj{k}{N}{\VF{\alpha}} - \Proj{k-2}{N}{\VF{\alpha}}\right) \proj{1}{\VF{\alpha}} \proj{2}{\VF{\alpha}} 
+ \sum_{k=0}^{N} \cml(k) \left(\Proj{k}{N}{\VF{\alpha}} - \Proj{k-1}{N}{\VF{\alpha}}\right) \left(\proj{1}{\VF{\alpha}} \qroj{2}{\VF{\alpha}} + \qroj{1}{\VF{\alpha}} \proj{2}{\VF{\alpha}}\right) \\
&+ \sum_{k=0}^{N} \cml(k) \Proj{k-2}{N-2}{\VF{\alpha}} \\
=&  \left(\dCml{-2}^{\VF{\alpha}}\right) \proj{1}{\VF{\alpha}} \proj{2}{\VF{\alpha}} + \left(\dCml{-1}^{\VF{\alpha}}\right) \left(\proj{1}{\VF{\alpha}} \qroj{2}{\VF{\alpha}} + \qroj{1}{\VF{\alpha}} \proj{2}{\VF{\alpha}}\right) + \sum_{k=0}^{N} \cml(k) \Proj{k-2}{N-2}{\VF{\alpha}}. \qedhere
\end{align*}
\end{proof}

\noindent The main idea of dealing with the unbounded interaction potential, which is linear in the $\xi$ variables, is to introduce a cutoff. This cutoff in the $\xi$ variables corresponds of course to regularity of the Fourier transform. Therefore, the following Lemma is essential to treat the large $\xi$ term.

\begin{lem}
\label{lem:mf:delta-approximation}
Let $\alpha_N\in\Wp{2,2}{\vec{\bz}}$ be $\Lp{2}{\vec{\bz}}$ normalized and symmetric, $\VF{\alpha}_N\equiv\fourier^{\otimes N}\alpha_N\in\Mp{2}{\vec{\bz}}$. Likewise, consider $\alpha\in\Wp{2,2}{\bz}$ and $\VF{\alpha}\equiv\fourier\alpha$. Assume that $\VF{\chi}:\RXi\to[0,1]$ is smooth, invariant under rotations, and $\VF{\chi}\left|_{\{\abs{\cdot}\leq 1\}}\right.=1$, $\VF{\chi}\left|_{\{\abs{\cdot}\geq 2\}}\right.=0$. Define $\chi\equiv\inv{\fourier}\VF{\chi}$ and for $R>0$ any $\VF{\chi}_R(\xi)\equiv \VF{\chi}(\frac{\xi}{R})$. This yields
\begin{align}
\label{eqn:mf:chi-approx-n}
\Lpn{\left(1-\VF{\chi}_R(\xi_1)\right) \xi_1 \VF{\alpha}_N}{2}{\vec{\VF{\bz}}} \leq&~ \frac{1}{R} \Lpn{\Del{2}_{\bz_1} \alpha_N}{2}{\vec{\bz}}, \\
\label{eqn:mf:chi-approx-1}
\Lpn{\left(1-\VF{\chi}_R(\xi_1)\right) \xi_1 \proj{1}{\VF{\alpha}} \VF{\alpha}_N}{2}{\vec{\VF{\bz}}} \leq&~ \frac{1}{R} \Lpn{\Del{2}_{\bz} \alpha}{2}{\bz}.
\end{align}
\end{lem}

\begin{proof}
Let $\alpha_N\in\Cp{2}{\vec{\bz}}\cap\Wp{2,2}{\vec{\bz}}$ be some function, then we have
\begin{align*}
\Lpn{\left(1-\VF{\chi}_R(\xi_1)\right) \xi_1 \VF{\alpha}_N}{2}{\vec{\VF{\bz}}}^2 =&~ \int_{\RZHN} \left(1-\VF{\chi}_R(\xi_1)\right)^2 \abs{\xi_1}^2 \abs{\VF{\alpha}_N(\vec{\VF{\bz}})}^2~ \intd\vec{\VF{\bz}} \\
\leq&~ \frac{1}{R^2} \int_{\RZHN} \abs{\xi_1}^4 \abs{\VF{\alpha}_N(\vec{\VF{\bz}})}^2~ \intd\vec{\VF{\bz}} = \frac{1}{R^2} \Lpn{\abs{\xi_1}^2 \VF{\alpha}_N}{2}{\vec{\VF{\bz}}} \leq \frac{1}{R^2} \Lpn{\Del{2}_{\bz_1}\alpha_N}{2}{\vec{\bz}}^2.
\end{align*}
For the second claim, we analogously compute
\begin{equation*}
\Lpn{\left(1-\VF{\chi}_R(\xi_1)\right) \xi_1 \proj{1}{\VF{\alpha}} \VF{\alpha}_N}{2}{\vec{\VF{\bz}}}^2 = \Lpn{\left(1-\VF{\chi}_R(\xi)\right) \xi~ \VF{\alpha}}{2}{\VF{\bz}}^2 \Lpn{\proj{1}{\VF{\alpha}} \VF{\alpha}_N}{2}{\vec{\VF{\bz}}}^2 \stackrel{\text{\eqref{eqn:mf:chi-approx-n} for $N=1$}}{\leq} \frac{1}{R^2} \Lpn{\Del{2}_{\bz}\alpha}{2}{\bz}^2. \qedhere
\end{equation*}
\end{proof}

\begin{lem}
\label{lem:mf:operator-norms}
Let $\VF\alpha\in\Mp{1}{\VF{\bz}}$ be some $\Lp{2}{\VF{\bz}}$ normalized function, $\alpha=\inv{\fourier}\VF{\alpha}$. Then we find for the operator norms
\begin{align}
\label{eqn:mf:v-p-p-bound}
\norm{\VF{V}(\VF{\bz}_1-\VF{\bz}_2) \proj{1}{\VF{\alpha}} \proj{2}{\VF{\alpha}}}_{\Lp{2}{\vec{\VF{\bz}}}\to\Lp{2}{\vec{\VF{\bz}}}} \leq&~ C_\Gamma \Mpn{\VF{\alpha}}{1}{\VF{\bz}}^2 = C_\Gamma \Wpn{\alpha}{1,2}{\vec{\bz}}^2 \\
\label{eqn:mf:v-mf-p-bound}
\norm{\convol{\VF{V}}{\abs{\VF{\alpha}^2}}(\VF{\bz}_1) \proj{1}{\VF{\alpha}}}_{\Lp{2}{\vec{\VF{\bz}}}\to\Lp{2}{\vec{\VF{\bz}}}} \leq&~ C_\Gamma \Mpn{\VF{\alpha}}{1}{\VF{\bz}}^2 = C_\Gamma \Wpn{\alpha}{1,2}{\bz}^2 \\
\label{eqn:mf:nabla-p-bound}
\norm{\xi_1\proj{1}{\VF{\alpha}}}_{\Lp{2}{\vec{\VF{\bz}}}\to\Lp{2}{\vec{\VF{\bz}}}} =&~ \Lpn{\abs{\xi}\VF{\alpha}}{2}{\VF{\bz}} = \Lpn{\nabla\alpha}{2}{\bz}.
\end{align}
\end{lem}

\begin{proof}
Let $\VF{\alpha}_N\in\Lp{2}{\vec{\VF{\bz}}}$ be arbitrary. Using the estimate $\abs{\xi_1-\xi_2}\leq\left(1+\abs{\xi_1}^2\right)^{\frac{1}{2}} \left(1+\abs{\xi_2}^2\right)^{\frac{1}{2}}$ one finds 
\begin{align*}
\Lpn{\VF{V}(\VF{\bz}_1-\VF{\bz}_2) \proj{1}{\VF{\alpha}} \proj{2}{\VF{\alpha}} \VF{\alpha}_N}{2}{\vec{\VF{\bz}}}^2 
=&~ \SP{\VF{\alpha}_N}{\proj{1}{\VF{\alpha}} \proj{2}{\VF{\alpha}} \abs{\VF{V}(\VF{\bz}_1-\VF{\bz}_2)}^2 \proj{1}{\VF{\alpha}} \proj{2}{\VF{\alpha}} \VF{\alpha}_N} \\
=&~ \Lpn{\convol{\abs{\VF{V}}^2}{\abs{\VF{\alpha}}^2} \abs{\VF{\alpha}}^2}{1}{\VF{\bz}} \Lpn{\proj{1}{\VF{\alpha}} \proj{2}{\VF{\alpha}} \VF{\alpha}_N}{2}{\vec{\VF{\bz}}}
\leq \Lpn{\nabla\Gamma}{\infty}{\bx}^2 \Mpn{\VF{\alpha}}{1}{\VF{\bz}}^4 \Lpn{\VF{\alpha}_N}{2}{\vec{\VF{\bz}}}^2,
\end{align*}
The other two claims hold similarly.
\end{proof}

\noindent Now all ingredients are prepared to prove our main result.

\newcommand{\projs}[1]{\proj{{#1}}{}}
\newcommand{\qrojs}[1]{\qroj{{#1}}{}}
\newcommand{\CO}[1]{\beta_N({#1})}

\begin{thm}[Hamilton Hartree/ Hamilton Vlasov mean field limit]
\label{thm:mf:hamilton-hartree-mf-limit}
Assume the regularity of the potential \eqref{eqn:mf:bounded-force} given in Section \ref{sec:mf:regular-global-well-posedness} to be valid. Let $\mathring{\VF{\alpha}}\in\Mp{2}{\VF{\bz}}$ be an $\Lp{2}{\VF{\bz}}$ normalized initial state with corresponding solution $\VF{\alpha}: \RNp \to \Mp{2}{\VF{\bz}}$ of \eqref{eqn:mf:hamilton-hartree}, $\alpha(t) \equiv \inv{\fourier}\VF{\alpha}(t)$ the conjugate solution of \eqref{eqn:mf:hamilton-vlasov}. Likewise, let $\VF{\alpha}_N:\RNp \to \Mp{2}{\vec{\VF{\bz}}}$ be the sequence of solutions of \eqref{eqn:mf:pseudo-qm} with $\Lp{2}{\vec{\VF{\bz}}}$ normalized initial datum $\VF\alpha_N(0)=\mathring{\VF{\alpha}}_N \in \Mp{2}{\vec{\VF{\bz}}}$, $\alpha_N(t)\equiv \inv{\fourier^{\otimes N}}\VF{\alpha}_N(t)$ solution of \eqref{eqn:mf:liouville-a}, s.t. for some $M\geq1$
\begin{align*}
\sup_N \Lpn{\Del{1}_{\bz_1}\mathring{\alpha}_N}{2}{\vec{\bz}} \leq M, \quad \sup_N \Lpn{\Del{2}_{(\bz_1,\bz_2)}\mathring{\alpha}_N}{2}{\vec{\bz}} \leq M, \quad \Wpn{\mathring{\alpha}}{2,2}{\bz} \leq M.
\end{align*}
Choose $0\leq\lambda<1$. Then there are continuous, monotonic increasing functions $B_{1,M}, B_{2,M}:\RNp\to\RNp$ both independent of $N$, s.t.
\begin{align*}
\abs{\ExHN{\cml^{\VF{\alpha}(t)}}{t} - \SP {\mathring{\VF{\alpha}}_N} {\cml^{\mathring{\VF{\alpha}}} \mathring{\VF{\alpha}}_N}} \leq&~ \int_{0}^{t} B_{1,M}(\tau) \ExHN{\cml^{\VF{\alpha}(\tau)}}{\tau}~ \intd\tau \\
&+ \left(N^{-\lambda}+N^{-\frac{1-\lambda}{4}}\right) ~ \int_{0}^{t} B_{2,M}(\tau)~ \intd\tau \\
\leq&~\left( N^{-\lambda}+N^{-\frac{1-\lambda}{4}}\right) \left(\int_{0}^{t} B_{2,M}(\tau)~ \intd\tau\right) \exp\left(\int_{0}^{t} B_{1,M}(\tilde{\tau})~ \intd\tilde{\tau}\right).
\end{align*}
\end{thm}

\begin{rmk}
\label{rmk:mf:singular-potentials}
Assuming that $\Lpn{\Del{2}_{\bz_1}\alpha_N(t)}{2}{\vec{\bz}}$ and $\Lpn{\nabla\alpha(t)}{\infty}{\bz}$ are bounded uniformly in $N$ and $t$, it is sufficient to assume that $\nabla\Gamma\in\Lp{2}{\bx}$ to prove validity of Theorem \ref{thm:mf:hamilton-hartree-mf-limit}. Note, that while proving the assumption on $\alpha_N(t)$ might be quite difficult for singular potentials, the assumption on $\alpha(t)$ can be proven for a certain class of sufficiently smooth initial states, even for Coulomb interaction \cite{neiss}.

Although these additional assumptions are not completely satisfactory, it is still remarkable that  they allow for a derivation of the effective description under the presence of interaction potentials with mild singularities.
\end{rmk}

\begin{proof}
Let $\VF{\alpha}$ and $\VF{\alpha}_N$ fulfill the conditions of the Theorem. For the sake of readability, we use the shortened notation
\begin{equation*}
\cml \equiv \cml^{\VF{\alpha}(t)}, \quad \dCml{l} \equiv \dCml{l}^{\VF{\alpha}(t)} \quad \VF{V}_{m,n} \equiv \VF{V}(\VF{\bz}_m-\VF{\bz}_n), \quad \bar{V}_m \equiv \convol{\VF{V}}{\abs{\VF{\alpha}(t)}^2}(\VF{\bz}_m), \quad \projs{m} \equiv \proj{m}{\VF{\alpha}(t)}.
\end{equation*}
By Lemma \ref{lem:mf:derivative}, $t\mapsto \CO{t} \equiv \ExHN{\cml^{\VF{\alpha}(t)}}{t}$ is differentiable and therefore
\begin{align*}
\partial_t\CO{t} =&~ \partial_t\ExHN{\cml}{t} \stackrel{\text{Lem.\ref{lem:mf:derivative}}}{=} N~ \Im\ExHN{\cml \left(\VF{V}_{1,2} - \bar{V}_1 - \bar{V}_2\right)}{t} \\
\stackrel{\text{Lem.\ref{lem:mf:m-decomposition}}}{=}& N~ \Im\ExHN{\left(\dCml{-2}\right) \projs{1} \projs{2} \left(\VF{V}_{1,2} - \bar{V}_1 - \bar{V}_2\right)}{t} \\
&+ 2N~ \Im\ExHN{\left(\dCml{-1}\right) \projs{1} \qrojs{2} \left(\VF{V}_{1,2} - \bar{V}_1 - \bar{V}_2\right)}{t} \\
&+ N~ \Im\ExHN{\left(\sum_{k=0}^{N} \cml(k) \Proj{k-2}{N-2}{\VF{\alpha}}\right) \left(\VF{V}_{1,2} - \bar{V}_1 - \bar{V}_2\right)}{t}.
\end{align*}
We remark that $\Im\ExHN{A}{t}=0$ for any symmetric operator $A$. As the two operators in the third term commute and are symmetric, their product is symmetric and the term cancels. Now, inserting $\id_{\MFS{N}}=\projs{1}\projs{2}+\projs{1}\qrojs{2}+\qrojs{1}\projs{2}+\qrojs{1}\qrojs{2}$ in both remaining terms, one obtains
\begin{align}
\partial_t\CO{t} 
=&~ N~ \Im\ExHN{\left(\dCml{-2}\right) \projs{1} \projs{2} \left(\VF{V}_{1,2} - \bar{V}_1 - \bar{V}_2\right) \projs{1} \projs{2}}{t} \tag{1} \\
&+ N~ \Im\ExHN{\left(\dCml{-2}\right) \projs{1} \projs{2} \left(\VF{V}_{1,2} - \bar{V}_1 - \bar{V}_2\right) \projs{1} \qrojs{2}}{t} \tag{2} \\
&+ N~ \Im\ExHN{\left(\dCml{-2}\right) \projs{1} \projs{2} \left(\VF{V}_{1,2} - \bar{V}_1 - \bar{V}_2\right) \qrojs{1} \projs{2}}{t} \tag{3} \\
&+ N~ \Im\ExHN{\left(\dCml{-2}\right) \projs{1} \projs{2} \left(\VF{V}_{1,2} - \bar{V}_1 - \bar{V}_2\right) \qrojs{1} \qrojs{2}}{t} \tag{4} \\
&+ 2N~ \Im\ExHN{\left(\dCml{-1}\right) \qrojs{1} \projs{2} \left(\VF{V}_{1,2} - \bar{V}_1 - \bar{V}_2\right) \projs{1} \projs{2}}{t} \tag{5} \\
&+ 2N~ \Im\ExHN{\left(\dCml{-1}\right) \qrojs{1} \projs{2} \left(\VF{V}_{1,2} - \bar{V}_1 - \bar{V}_2\right) \projs{1} \qrojs{2}}{t} \tag{6} \\
&+ 2N~ \Im\ExHN{\left(\dCml{-1}\right) \qrojs{1} \projs{2} \left(\VF{V}_{1,2} - \bar{V}_1 - \bar{V}_2\right) \qrojs{1} \projs{2}}{t} \tag{7} \\
&+ 2N~ \Im\ExHN{\left(\dCml{-1}\right) \qrojs{1} \projs{2} \left(\VF{V}_{1,2} - \bar{V}_1 - \bar{V}_2\right) \qrojs{1} \qrojs{2}}{t}. \tag{8} 
\end{align}
Along with Lemma \ref{lem:mf:computation-rules}-(iv), we see that
\begin{align*}
\text{(1)}=&~ N\Im\ExHN{\left(\dCml{-2}\right) \projs{1} \projs{2} \left(\VF{V}_{1,2} - \bar{V}_1 - \bar{V}_2\right) \projs{1} \projs{2}}{t} \\
=&~ N\Im\ExHN{\projs{1} \projs{2} \left(\VF{V}_{1,2} - \bar{V}_1 - \bar{V}_2\right) \projs{1} \projs{2} \left(\dCml{-2}\right)}{t}^* \\
=& -N\Im\ExHN{\left(\dCml{-2}\right) \projs{1} \projs{2} \left(\VF{V}_{1,2} - \bar{V}_1 - \bar{V}_2\right) \projs{1} \projs{2}}{t} = 0,
\end{align*}
and by the very same argument also terms (6) and (7) vanish. Acknowledging the invariance w.r.t. particle permutation, we can also recombine terms (2), (3), yielding
\begin{align*}
\partial_t\CO{t} 
=&~ 2N~ \Im\ExHN{\left(\dCml{-2}\right) \projs{1} \projs{2} \left(\VF{V}_{1,2} - \bar{V}_1 - \bar{V}_2\right) \qrojs{1} \projs{2}}{t} \tag{1'} \\
&+ N~ \Im\ExHN{\left(\dCml{-2}\right) \projs{1} \projs{2} \left(\VF{V}_{1,2} - \bar{V}_1 - \bar{V}_2\right) \qrojs{1} \qrojs{2}}{t} \tag{2'} \\
&+ 2N~ \Im\ExHN{\left(\dCml{-1}\right) \qrojs{1} \projs{2} \left(\VF{V}_{1,2} - \bar{V}_1 - \bar{V}_2\right) \projs{1} \projs{2}}{t} \tag{3'} \\
&+ 2N~ \Im\ExHN{\left(\dCml{-1}\right) \qrojs{1} \projs{2} \left(\VF{V}_{1,2} - \bar{V}_1 - \bar{V}_2\right) \qrojs{1} \qrojs{2}}{t}. \tag{4'}
\end{align*}
Finally, we compute for the third term (3') after interchanging variable indices 1 and 2
\begin{align*}
\text{(3')} =&~ -2N~ \Im\ExHN{\projs{1} \projs{2} \left(\VF{V}_{1,2} - \bar{V}_1 - \bar{V}_2\right) \qrojs{1} \projs{2} \left(\dCml{-1}\right)}{t} \\
=&~ -2N~ \Im\ExHN{\shift{\dCml{-1}}{-1} \projs{1} \projs{2} \left(\VF{V}_{1,2} - \bar{V}_1 - \bar{V}_2\right) \qrojs{1} \projs{2}}{t},
\end{align*}
which recombines with term (1') and gives
\begin{align*}
\partial_t\CO{t} 
=&~ 2N~ \Im\ExHN{\left(\dCml{-1}\right) \projs{1} \projs{2} \left(\VF{V}_{1,2} - \bar{V}_1 - \bar{V}_2\right) \qrojs{1} \projs{2}}{t} \\
&+ N~ \Im\ExHN{\left(\dCml{-2}\right) \projs{1} \projs{2} \left(\VF{V}_{1,2} - \bar{V}_1 - \bar{V}_2\right) \qrojs{1} \qrojs{2}}{t} \\
&+ 2N~ \Im\ExHN{\left(\dCml{-1}\right) \qrojs{1} \projs{2} \left(\VF{V}_{1,2} - \bar{V}_1 - \bar{V}_2\right) \qrojs{1} \qrojs{2}}{t}.
\end{align*}
The first term exactly vanishes, because $\projs{1}\bar{V}_2\qrojs{1}=\projs{1}\qrojs{1}\bar{V}_2=0$ and $\projs{2}\VF{V}_{1,2}\projs{2}=\projs{2}\bar{V}_1\projs{2}$, where the first equality can be used also to simplify the other two expressions by omitting $\bar{V}_1,\bar{V}_2$ and $\bar{V}_1$ respectively, resulting in
\begin{align*}
\partial_t\CO{t} 
=&~ N~ \Im\ExHN{\left(\dCml{-2}\right) \projs{1} \projs{2} \VF{V}_{1,2} \qrojs{1} \qrojs{2}}{t} \tag{1''} \\
&+ 2N~ \Im\ExHN{\left(\dCml{-1}\right) \qrojs{1} \projs{2} \left(\VF{V}_{1,2} - \bar{V}_2\right) \qrojs{1} \qrojs{2}}{t}. \tag{2''}
\end{align*}
Estimating the absolute value of (1'') we use  the notation $\bar{\cn}(k)\equiv \sqrt{\frac{N}{k}}$ for $1\leq k\leq N$ and $\bar{\cn}(k)=0$ else and observe that  $\bar{\cn}\cn=\id-\Proj{0}{N}{}$  is \textit{almost} the inverse of $\cn$. Following this idea we compute
\begin{align*}
\abs{\text{(1'')}} =&~ N~ \abs{\ExHN{\left(-\dCml{-2}\right)^{\frac 12} \cn \projs{1} \projs{2} \VF{V}_{1,2} \shift{\left(-\dCml{-2}\right)^{\frac 12}}{2} \shift{\bar{\cn}}{2} \qrojs{1} \qrojs{2}}{t}} \\
\leq&~ N~ \Lpn{\left(-\dCml{-2}\right)^{\frac 12} \cn~\VF{\alpha}_N(t)}{2}{\vec{\VF{\bz}}} \norm{\VF{V}_{1,2} \projs{1} \projs{2}}_\text{op} \Lpn{\shift{\left(-\dCml{-2}\right)^{\frac 12}}{2} \shift{\bar{\cn}}{2} \qrojs{1} \qrojs{2} \VF{\alpha}_N(t)}{2}{\vec{\VF{\bz}}} \\
\leq&~ N~ \left(\frac{2(1+2N^{-\lambda})}{N}\right)^{\frac 12} \Lpn{\left(\cml\right)^{\frac 12}~ \VF{\alpha}_N(t)}{2}{\vec{\VF{\bz}}} \norm{\VF{V}_{1,2} \projs{1} \projs{2}}_{\text{op}} \\
&\cdot \left(\frac{N}{N-1}\right)^{\frac 12} \Lpn{\shift{\left(-\dCml{-2}\right)^{\frac 12}}{2} \shift{\bar{\cn}}{2} \cn^2~ \VF{\alpha}_N(t)}{2}{\vec{\VF{\bz}}} \\
\leq&~N \left(\frac{2(1+N^{-\lambda})}{N-1}\right)^{\frac 12}   \left(\frac{2}{N}\right)^{\frac 12} \norm{\VF{V}_{1,2} \projs{1} \projs{2}}_{\text{op}} \Lpn{\left(\cml\right)^{\frac 12} \VF{\alpha}_N(t)}{2}{\vec{\VF{\bz}}}\left(\Lpn{\left(\cml\right)^{\frac 12} \VF{\alpha}_N(t)}{2}{\vec{\VF{\bz}}}+2N^{-\frac{\lambda}{2}}\right) \\
\stackrel{\text{Lem.\ref{lem:mf:operator-norms}}}{\leq}&~ \left(4\frac{1+N^{-\lambda}}{1-N^{-1}}\right)^{\frac{1}{2}} C_\Gamma\Wpn{\alpha(t)}{1,2}{\bz}^2 (2\CO{t}+N^{-\lambda}) \\
\stackrel{\text{Cor.\ref{cor:mf:derivative-bound}}}{\leq}&~  2\left(\frac{1+N^{-\lambda}}{1-N^{-1}}\right)^{\frac{1}{2}} C_\Gamma \left(\bound{M}{\Wp{1,2}{\bz}}(t)\right)^2 (2\CO{t}+N^{-\lambda}) .
\end{align*}
In order to estimate (2''), one picks a smooth cutoff $\chi$ with the properties of Lemma \ref{lem:mf:delta-approximation}, regrouping $\VF{V}_{1,2}=-\nabla\Gamma_{1,2} \cdot (\VF{\chi}_R(\xi_1) \xi_1-\xi_2) - \nabla\Gamma_{1,2} \cdot (1-\VF{\chi}_R(\xi_1)) \xi_1$, \begin{align*}
\abs{\text{(2'')}} \leq&~ 2N~ \abs{\ExHN{\left(\dCml{-1}\right) \qrojs{1} \projs{2} \left(\bar{V}_2 + \nabla\Gamma_{1,2} \cdot (\VF{\chi}_R(\xi_1) \xi_1-\xi_2)\right) \qrojs{1} \qrojs{2}}{t}} \tag{1'''} \\
&+ 2N~ \abs{\ExHN{\left(\dCml{-1}\right) \qrojs{1} \projs{2} \left(\nabla\Gamma_{1,2} \cdot (1-\VF{\chi}_R(\xi_1)) \xi_1\right) \qrojs{1} \qrojs{2}}{t}}. \tag{2'''}
\end{align*}
As the multiplication operator in (1''') times $\projs{2}$ is bounded, because with help of Lemma \ref{lem:mf:operator-norms}, we find
\begin{align*}
\norm{\left(\bar{V}_2+\nabla\Gamma_{1,2} \cdot (\VF{\chi}_R(\xi_1)\xi_1 - \xi_2)\right) \projs{2}}_{\text{op}} \leq&~ \norm{\bar{V}_2\projs{2}}_{\text{op}} + \Lpn{\nabla\Gamma}{\infty}{\bx} \left(\Lpn{\VF{\chi}_R(\xi_1)\xi_1}{\infty}{\VF{\bz}} + \norm{\xi_2\projs{2}}_{\text{op}}\right) \\
\leq&~ C_\Gamma \left(\Wpn{\alpha(t)}{1,2}{\bz}^2 + 2R + \Lpn{\nabla\alpha(t)}{2}{\bz}\right).
\end{align*}
This can be used to estimate
\begin{align*}
\text{(1''')} =&~ 2N~ \abs{\ExHN{\left(-\dCml{-1}\right)^{\frac{1}{2}} \qrojs{1} \projs{2} \bar{V}_2 \shift{\left(-\dCml{-1}\right)^{\frac 12}}{1} \qrojs{1} \qrojs{2}}{t}} \\
&+ 2N~ \abs{\ExHN{\left(-\dCml{-1}\right)^{\frac{1}{2}} \qrojs{1} \projs{2} \nabla\Gamma_{1,2} \cdot (\VF{\chi}_R(\xi_1) \xi_1-\xi_2) \shift{\left(-\dCml{-1}\right)^{\frac 12}}{1} \qrojs{1} \qrojs{2}}{t}} \\
\leq&~ 2N~ \Lpn{\left(-\dCml{-1}\right)^{\frac{1}{2}} \qrojs{1} \VF{\alpha}_N(t)}{2}{\vec{\VF{\bz}}} C_\Gamma \left(\Wpn{\alpha(t)}{1,2}{\bz}^2 + 2R + \Lpn{\nabla\alpha(t)}{2}{\bz}\right) \\
&~ \quad \cdot \Lpn{\shift{\left(-\dCml{-1}\right)^{\frac{1}{2}}}{1} \qrojs{1} \qrojs{2} \VF{\alpha}_N(t)}{2}{\vec{\VF{\bz}}} \\
\leq&~ 2~ \left(1+N^{-\lambda}\right)^{\frac 12} \left(\frac{N}{N-1}\right)^{\frac 12} \left(\sup_{k\leq N^\lambda} \frac{k}{N}\right)^{\frac 12} C_\Gamma \left(\Wpn{\alpha(t)}{1,2}{\bz}^2 + 2R + \Lpn{\nabla\alpha(t)}{2}{\bz}\right) \Lpn{\cml^{\frac 12} \VF{\alpha}_N(t)}{2}{\vec{\VF{\bz}}}^2 \\
\stackrel{\text{Cor.\ref{cor:mf:derivative-bound}}}{\leq}&~ 4N^{\frac{\lambda-1}{2}} \left(\frac{1+N^{-\lambda}}{1-N^{-1}}\right)^{\frac 12} C_\Gamma \left(R + \left(\bound{M}{\Wp{1,2}{\bz}}(t)\right)^2\right) \CO{t}.
\end{align*}
At last, term (2''') yields with Lemmas \ref{lem:mf:computation-rules}-(iv), \ref{lem:mf:delta-approximation}:
\begin{align*}
\text{(2''')} \leq&~ 2N~ \Lpn{\nabla\Gamma}{\infty}{\bx} \Lpn{(1-\VF{\chi}_R(\xi_1)) \xi_1 (1-\projs{1}) \VF{\alpha}_N(t)}{2}{\vec{\VF{\bz}}} \Lpn{\shift{\dCml{-1}}{1} \qrojs{1} \qrojs{2} \VF{\alpha}_N(t)}{2}{\vec{\VF{\bz}}} \\
\leq&~ 2~ C_\Gamma \left(\frac{N}{N-1}\right)^{\frac{1}{2}} \frac{1}{R} \left(\Lpn{\Del{2}_{\bz_1}\alpha_N(t)}{2}{\vec{\bz}} + \Lpn{\Del{2}_{\bz}\alpha(t)}{2}{\bz}\right) N\Lpn{\shift{\dCml{-1}}{1} \cn^2 \VF{\alpha}_N(t)}{2}{\vec{\VF{\bz}}} \\
\stackrel{\text{Lem.\ref{lem:mf:n-derivative-bounds},Cor.\ref{cor:mf:2nd-derivative-bound}}}{\leq}&~ 2~ C_\Gamma \left(\frac{1}{1-N^{-1}}\right)^{\frac 12} \frac{1}{R} \left(\bound{M,\Del{2}_{(\bz_1,\bz_2)}}{\Lp{2}{\vec{\bz}}}(t) + \bound{M,\Del{2}}{\Lp{2}{\bz}}(t)\right).
\end{align*}
Combining all the estimates for (1''), (1'''), and (2'''), yields the Gronwall type estimate
\begin{align*}
\abs{\partial_t\CO{t}} \leq&~ 2~ C_\Gamma \left(\frac{1+N^{-\lambda}}{1-N^{-1}}\right)^{\frac 12} \left(2 \left(\bound{M}{\Wp{1,2}{\bz}}(t)\right)^2 + 2N^{\frac{\lambda-1}{2}} \left(R+\left(\bound{M}{\Wp{1,2}{\bz}}(t)\right)^2\right)\right)~ \CO{t} \\
&+2~ C_\Gamma \left(\frac{1}{1-N^{-1}}\right)^{\frac 12} \left(\bound{M}{\Wp{1,2}{\bz}}(t)\right)^2N^{-\lambda}+ 2~ C_\Gamma \left(\frac{1}{1-N^{-1}}\right)^{\frac 12} \left(\bound{M,\Del{2}_{(\bz_1,\bz_2)}}{\Lp{2}{\vec{\bz}}}(t) +  \bound{M,\Del{2}}{\Lp{2}{\bz}}(t)\right) \frac{1}{R}.
\end{align*}
If we choose $R(N)\equiv N^{\delta}$, $0<\delta\leq\frac{1-\lambda}{2}$ and optimize the choice of $\delta$, we obtain $\delta=\frac{1-\lambda}{4}$ and the convergence rate of all error terms is $\mathcal{O}\left(N^{-{\frac{1-\lambda}{4}}}\right)+\mathcal{O}\left(N^{-\lambda}\right)$.
\end{proof}

\section*{Conflict of interest statement}

The authors declare that there are no conflicts of interest, because this work has not been funded by third parties.

\bibliographystyle{abbrv}

\end{document}